\pdfoutput=1
\newif\iffull
\fulltrue

\documentclass{article}
\usepackage{amssymb, amsmath, amsthm}

\usepackage{fullpage}
\usepackage[english]{babel} 
\usepackage[utf8]{inputenc} 
\usepackage{graphicx}
\usepackage{cleveref}
\usepackage{caption}
\usepackage{subcaption}
\usepackage{comment}

\title{Classifying Convex Bodies by their\\Contact and Intersection Graphs}
\author{Anders Aamand\footnote{Basic Algorithms Research Copenhagen (BARC), University of Copenhagen.
BARC is supported by the VILLUM Foundation grant 16582.}\hspace{35pt}
Mikkel Abrahamsen$^\ast$\\
Jakob Bæk Tejs Knudsen$^\ast$\hspace{35pt}
Peter Michael Reichstein Rasmussen$^\ast$}
\date{\today}

\newcommand{\abs}[1]{\left\lvert #1\right\rvert}

\newcommand{\setbuilder}[2]{\left\{ #1 \; \middle\vert \; #2 \right\}}
\newcommand{\norm}[2]{\left\| #1 \right\|_{#2}}
\newcommand{\ball}[2]{B_{#2}(#1)}
\newcommand{\conv}{\mathop{\mathrm{conv}}}

\newcommand{\R}{\mathbb{R}}
\newcommand{\N}{\mathbb{N}}
\newcommand{\Z}{\mathbb{Z}}

\newcommand{\eps}{\varepsilon}
\newcommand{\mydef}{:=}

\newtheorem{theorem}{Theorem}
\newtheorem{lemma}[theorem]{Lemma}

\newtheorem{corollary}[theorem]{Corollary}
\newtheorem{proposition}[theorem]{Proposition}

\theoremstyle{definition}
\newtheorem{definition}[theorem]{Definition}
\newtheorem{construction}[theorem]{Construction}
\newtheorem{remark}[theorem]{Remark}

\begin{document}
	\maketitle
	
	\begin{abstract}
Suppose that $A$ is a convex body in the plane and that $A_1,\dots,A_n$ are translates of $A$. Such translates give rise to an \emph{intersection graph} of $A$, $G=(V,E)$, with vertices $V=\{1,\dots,n\}$ and edges $E=\{uv\mid A_u\cap A_v\neq \emptyset\}$. The subgraph $G'=(V, E')$ satisfying that $E'\subset E$ is the set of edges $uv$ for which the interiors of $A_u$ and $A_v$ are disjoint is a \emph{unit distance} graph of $A$. If furthermore $G'=G$, i.e., if the interiors of $A_u$ and $A_v$ are disjoint whenever $u\neq v$, then $G$ is a \emph{contact graph} of $A$.

In this paper we study which pairs of convex bodies have the same contact, unit distance, or intersection graphs.
We say that two convex bodies $A$ and $B$ are equivalent if there exists a linear transformation $B'$ of $B$ such that for any slope, the longest line segments with that slope contained in $A$ and $B'$, respectively, are equally long.
For a broad class of convex bodies, including all strictly convex bodies and linear transformations of regular polygons, we show that the contact graphs of $A$ and $B$ are the same if and only if $A$ and $B$ are equivalent.
We prove the same statement for unit distance and intersection graphs.


\end{abstract}

	\section{Introduction}

Consider a convex body $A$, i.e., a convex, compact region of the plane with non-empty interior, and let $\mathcal A=\{A_1,\ldots,A_n\}$ be a set of $n$ translates of $A$.
Then $\mathcal A$ gives rise to an \emph{intersection graph} $G=(V,E)$, where $V=\{1,\ldots,n\}$ and $E=\{uv\mid A_u\cap A_v\neq\emptyset\}$, and a \emph{unit distance graph} $G'=(V,E')$, where $uv\in E'$ if and only if $uv\in E$ and $A_u$ and $A_v$ have disjoint interiors.
In the special case that $G=G'$ (i.e., the convex bodies of $\mathcal A$ have pairwise disjoint interiors), we say that $G$ is a \emph{contact graph} (also known as a \emph{touch graph} or \emph{tangency graph}).
Thus, $A$ defines three classes of graphs, namely the intersection graphs $I(A)$, the unit distance graphs $U(A)$, and the contact graphs $C(A)$ of translates of $A$.

The study of intersection graphs has been an active research area in discrete and computational geometry for the past three decades.
Numerous papers consider the problem of solving classical graph problems efficiently on various classes of geometric intersection graphs.
From a practical point of view, the research is often motivated by the applicability of intersection graphs when modeling wireless communication networks and facility location problems.
If a station is located at some point in the plane and is able to transmit to and receive from all other stations within some distance then the stations can be represented as disks in such a way that two stations can communicate if and only if their disks overlap.

Meanwhile, the study of contact graphs of translates of a convex body has older roots.  It is closely related to the packings of such a body, which has a very long and rich history in mathematics going back (at least) to the seventeenth century, where research on the packings of circles of varying and constant radii was conducted and Kepler famously conjectured upon a 3-dimensional counterpart of such problems, the packing of spheres.
An important notion in this area is that of the \emph{Hadwiger number} of a body $K$, which is the maximum possible number of pairwise interior-disjoint translates $K_i$ of $K$ that each touch but do not overlap $K$.
The Hadwiger number of $K$ is thus the maximum degree of a contact graph of translates of $K$.
In the plane, the Hadwiger number is $8$ for parallelograms and $6$ for all other convex bodies.
We refer the reader to the books and surveys by L{\'a}szl{\'o} and G{\'a}bor Fejes T{\'o}th~\cite{toth1983new,toth1972lagerungen} and B{\"o}r{\"o}czky~\cite{boroczky2004finite}.

Another noteworthy result on contact graphs is the Circle Packing Theorem (also known as the Koebe--Andreev--Thurston Theorem):
A graph is simple and planar if and only if it is the contact graph of some set of circular disks in the plane (the radii of which need not be equal).
The result was proven by Koebe in 1935~\cite{koebe1936kontaktprobleme} (see~\cite{felsner2018primal} for a streamlined, elementary proof).
Schramm~\cite{schramm2007combinatorically} generalized the circle packing theorem by showing that if a smooth planar convex body is assigned to each vertex in a planar graph, then the graph can be realized as a contact graph where each vertex is represented by a homothet (i.e., a scaled translation) of its assigned body.

In this paper we investigate the question of when two convex bodies $A$ and $B$ give rise to the same classes of graphs.
We restrict ourselves to convex bodies $A$ that have the URTC property (\emph{unique regular triangle constructibility}). This is the property that given two interior disjoint translates $A_1,A_2$ of $A$ that touch, there are exactly two ways to place a third translate $A_3$ such that $A_3$ is interior disjoint from $A_1$ and $A_2$, but touches both.
Convex bodies with the URTC property include all linear transformations of regular polygons except squares and all strictly convex bodies~\cite{geher2015contribution}.

The main result of the paper is summarized in the following theorem.
\begin{theorem}\label{thm:main}
Let $A$ and $B$ be convex bodies with the URTC property. Then each of the identities $I(A)=I(B)$, $U(A)=U(B)$, and $C(A)=C(B)$ holds if and only if the following condition is satisfied:
there is a linear transformation $B'$ of $B$ such that for any slope, the longest segments contained in $A$ and $B'$, respectively, with that slope are equally long.
\end{theorem}

\subsection{Other Related Work}
Several papers have compared classes of intersection graphs of various geometric objects, see for instance~\cite{cabello2017refining,cardinal2017intersection,chaplick2014grid,janson1992thresholds,kratochvil1994intersection}.
Most of the results are inclusions between classes of intersection graphs of one-dimensional objects such as line segments and curves.

A survey by Swanepoel~\cite{swanepoel2018combinatorial} summarizes results on minimum distance graphs and unit distance graphs in normed spaces, including bounds on the minimum/maximum degree, maximum number of edges, chromatic number, and independence number.

Perepelitsa~\cite{perepelitsa2003bounds} studied unit disk intersection graphs in normed planar spaces and showed that they are $\chi$-bounded in any such space.
Kim et al.~\cite{kim2004chromatic} improved Perepelitsa's bound.
For other work on intersection graphs of translates of a fixed convex body, see~\cite{dumitrescu2011piercing,dumitrescu2012coloring, kim2008coloring,kim2006transversal}.

In the area of computational geometry, M{\"{u}}ller et al.~\cite{muller2013integer} gave sharp upper and lower bounds on the size of an integer grid used to represent an intersection graph of translates of a convex polygon with corners at rational coordinates.
Their results imply that for any convex polygon $R$ with rational corners, the problem of recognizing intersection graphs of translates of $R$ is in \textsc{NP}.
On the contrary, it is open whether recognition of unit disk graphs in the Euclidean plane is in NP.
Indeed, the problem is $\exists\R$-complete (and thus in PSPACE), and using integers to represent the center coordinates and radii of the disks in some graphs requires exponentially many bits~\cite{cardinal2015computational,mcdiarmid2013integer}.

	\subsection{Preliminaries}
We begin by defining some basic geometric concepts and terminology. 
\iffull
\subsubsection{Convex Bodies and Graphs as Point Sets}
\fi
For a subset $A\subset \R^2$ of the plane we denote by $A^\circ$ the interior of $A$.
We say that $A$ is a \emph{convex body} if $A$ is compact, convex, and has non-empty interior. 
We say that $A$ is \emph{symmetric} if whenever $x\in A$, then $-x\in A$. 
It is well-known that if $A$ is a symmetric convex body, then the map $\norm{\cdot}{A}:\R^2 \to \R_{\geq 0}$ defined by
$$
\norm{x}{A}=\inf\{\lambda\geq 0 \mid x \in \lambda A\},
$$
is a norm. 
Moreover $A=\{x\in \R^2\mid \norm{x}{A}\leq 1\}$ and $A^\circ=\{x\in \R^2\mid \norm{x}{A}< 1\}$.

It follows from these properties that for translates $A_1=A+v_1$ and $A_2=A+v_2$ it holds that $A_1\cap A_2\neq \emptyset$ if and only if $\norm{v_1-v_2}{A}\leq 2$ and $A_1^\circ \cap A_2^\circ \neq \emptyset$ if and only if $\norm{v_1-v_2}{A}<2$. 
This means that when studying contact, unit distance, and intersection graphs of a symmetric convex body $A$, we can shift viewpoint from translates of $A$ to point sets in $\R^2$ and their $\norm{\cdot}{A}$-distances:
If $\mathcal{A}\subset \R^2$ is a set of points we define $I_A(\mathcal{A})$ and $U_A(\mathcal{A})$  to be the graphs with vertex set $\mathcal{A}$ and edge sets $\{(x,y)\in \mathcal{A}^2 \mid x \neq y \text { and } \|x-y\|_A\leq 2\}$ and $\{(x,y)\in \mathcal{A}^2 \mid x \neq y \text { and } \|x-y\|_A= 2\}$, respectively.
Moreover, if for all distinct points $x,y\in \mathcal{A}$ it holds that $\|x-y\|_A\geq 2$, we say that $\mathcal{A}$ is \emph{compatible with $A$} and define $C_A(\mathcal{A})$ to be the graph with vertex set $\mathcal{A}$ and edge set $\{(x,y)\in \mathcal{A}^2 \mid x \neq y \text { and } \|x-y\|_A= 2\}$.
Then $I_A(\mathcal{A}), U_A(\mathcal{A})$, and $C_A(\mathcal{A})$, respectively, are isomorphic to the intersection, unit distance, and contact graph of $A$ realized by the translates $(A+a)_{a\in\mathcal{A}}$.
When studying contact, unit distance, and intersection graphs of a symmetric, convex body $A$ we will view them as being induced by point sets rather than by translates of $A$.

\iffull
\subsubsection{The URTC Property}
\fi
We say that a (not necessarily symmetric) convex body $A$ in the plane has the URTC property if the following holds: For any two interior disjoint translates of $A$, call them $A_1$ and $A_2$, satisfying that $A_1\cap A_2\neq \emptyset$, there exists precisely two vectors $v\in \R^2$ such that for $i\in \{1,2\}$, $(A+v)^\circ \cap A_i^\circ=\emptyset$ but $(A+v) \cap A_i\neq\emptyset$. If $A$ is symmetric, this amounts to saying that for any two points $v_1,v_2\in \R^2$ with $\norm{v_1-v_2}{A}=2$, the set $\{v\in \R^2 \mid \norm{v-v_1}{A}=\norm{v-v_2}{A}=2\}$ has size two.
Geh{\'e}r~\cite{geher2015contribution} proved that a symmetric convex body $A$ has the URTC property if and only if the boundary $\partial A$ does not contain a line segment of length more than~$1$ in the $\norm{\cdot}{A}$-norm.

\iffull
\subsubsection{Drawing of a Graph}
\fi
A \emph{drawing} of a graph $G \in I(A)$ as an intersection graph of a convex body
$A$ is a point set $\mathcal{A} \subset \R^2$ and a set of straight line segments
$\mathcal{L}$ such that $I_A(\mathcal{A})$ is isomorphic to $G$ and $\mathcal{L}$
is exactly the line segments between the points $u, v \in \mathcal{A}$
which are connected by an edge in $G$.
We define a drawing of a graph $G$ as a contact
and unit distance graph similarly.

\iffull
\subsubsection{Notation}
\fi
For a norm $\norm{\cdot}{}$ on $\R^2$ and a line segment $\ell$ with endpoints $a$ and $b$ we will often write $\norm{\ell}{}=\norm{ab}{}$ instead of $\norm{a-b}{}$. Also, if $A$ is a symmetric convex body and $U,V\subset \R^2$, we define $d_A(U,V):=\inf\{\norm{uv}{A}\mid (u,v)\in U \times V\}$.

	\subsection{Structure of the Paper}\label{papstruc}
\iffull
In Section~\ref{suffapp}, we establish the sufficiency of the condition of Theorem~\ref{thm:main}, which is relatively straightforward.
\fi

\iffull
In Section~\ref{symmetriapp} we show how to reduce Theorem~\ref{thm:main} to the case where the convex bodies are symmetric.
For contact graphs, we then prove the following more general version of the necessity of the condition of Theorem~\ref{thm:main} in Section~\ref{necsec}.
\fi
\begin{theorem}\label{main:contact}
Let $A$ and $B$ be symmetric convex bodies with the URTC property such that $A$ is not a linear transformation of $B$. There exists a graph $G\in C(A)$ such that for all $H\in C(B)$ and all subgraphs $H'\subset H$, $G$ is not isomorphic to $H'$. In particular $C(A) \setminus C(B) \neq \emptyset$.
\end{theorem}
As we will also discuss in Section~\ref{necsec} the same result holds if $C(X)$ is replaced by $U(X)$ for $X\in \{A,B\}$ everywhere in the theorem above. The proof is identical. 

In Section~\ref{intsec} we prove the following result which combined with~Theorem~\ref{main:contact} yields the necessity of the condition of Theorem~\ref{thm:main} for \emph{intersection graphs}.
\begin{theorem}\label{theorem:main}
	Let $A$ and $B$ be symmetric convex bodies. If there exists a graph $G\in C(A)$ such
	that for all $H \in C(B)$ and all subgraphs $H'\subset H$, $G$
	is not isomorphic to $H'$, then $I(A) \neq I(B)$.
\end{theorem}
This result holds for general symmetric convex bodies.
An improvement of Theorem~\ref{main:contact} to general symmetric convex bodies (not necessarily having the URTC property) would thus yield a version of Theorem~\ref{thm:main} that also holds for general convex bodies.
	\section{Sufficiency of the Condition of Theorem \ref{thm:main}}\label{suffapp}
This section establishes the sufficiency of the condition stated in Theorem~\ref{thm:main} -- this is the easy direction. It is worth noting that for this direction our result holds for general convex bodies.

Essentially, we show that the classes of contact, unit distance, and intersection graphs arising from a convex body $A$ are closed under linear transformations of $A$ and under operations on $A$ maintaining the \emph{signature} of $A$ which we proceed to define.
\begin{definition}[Profile]
	Let $A\subset \R^2$ be a convex body. A \emph{profile} through $A$ at angle $\theta\in [0, \pi)$ is a closed line segment $\ell_\theta$ of maximal length which has argument $\theta$ and is contained in $A$. 
\end{definition}
\begin{definition}[Signature]
	Let $A\subset \R^2$ be a convex body. The \emph{signature} of $A$ is the function $\rho_A\colon [0, \pi)\to \R$ satisfying that for every $\theta\in [0, \pi)$, $\rho_A(\theta)=\abs{\ell_\theta}$ is the length of a profile through $A$ at angle $\theta$.
\end{definition}

\begin{lemma}\label{Lem-Moving_with_argument}
	Let $A\subset \R^2$ be a convex body and $v\in \R^2$ a vector with argument $\theta\in [0, 2\pi)$ and magnitude $r$. Then 
	\begin{enumerate}
		\item $(A+v)\cap A \neq \emptyset$ if and only if $r\leq \rho_A(\theta)$.
		\item $(A+v)^{\circ}\cap A^\circ\neq\emptyset$ if and only if $r< \rho_A(\theta)$.
	\end{enumerate}
\end{lemma}
\begin{proof}
		First, suppose that $a\in (A+v)\cap A$. Then $a-v\in A$ and by convexity, the line segment $\ell$ from $a-v$ to $a$, which has length $r$ and argument $\theta$, is contained in $A$. It follows that $r\leq \rho_A(\theta)$. If further $a\in A^{\circ}$, there exists a vector $u$ of length $\epsilon>0$ and argument $\theta$ such that $a+u\in A$. It follows that $(A+v+u)\cap A\neq \emptyset$ implying that $r+\epsilon\leq \rho_A(\theta)$, so $r<\rho_A(\theta)$.

		Second, suppose $r\leq \rho_A(\theta)$ and let $a, b$ be the endpoints of a profile, $\ell_\rho$, through $A$ at angle $\theta$ such that the vector from $a$ to $b$ has argument $\theta$. Then $a+v\in \ell_\theta$ as $r\leq \rho_A(\theta)=\abs{\ell_\theta}$. It follows that $a+v\in A\cap (A+v)$ implying $A\cap (A+v)\neq \emptyset$. If further $r<\rho_A(\theta)$, then $a+v\neq b$. Let $c\in A$ be such that $a,b$ and $c$ are not colinear. (Such a point exists as $A$ has non-empty interior.) The interiors of the triangles with vertices $a,b,c$ and $a+v,b+v,c+v$ are contained in $A^\circ$ and $(A+v)^\circ$, respectively and their intersection is non-trivial. It follows that $A^\circ \cap (A+v)^\circ\neq \emptyset$, as desired.

\end{proof}
Second, the following lemma demonstrates closure of contact, unit distance, and intersection graphs under linear transformations.
\begin{lemma}\label{Lem-Graphs_invariant_under_lt}
	Let $A\subset \R^2$ be a compact and connected subset of the plane and $M\colon \R^2\to\R^2$ be an invertible linear transformation. Then $C(A)=C(M(A))$, $U(A)=U(M(A))$, and $I(A)=I(M(A))$.
\end{lemma}
\begin{proof}
    Suppose that $A_1, \dots, A_n$ is a realization of a graph $G$ as either a contact, unit distance, or intersection graph of $A$. Then $M(A_1), \dots, M(A_n)$ is a realization of $G$ as a contact, unit distance, or intersection graph, respectively, of $M(A)$.
\end{proof}
Finally, we combine the above observations to prove our sufficient condition.
\begin{theorem}\label{suff:thm}
	Let $A, B\subset \R^2$ be convex bodies. If there exists a linear transformation $M\colon \R^2\to\R^2$ such that the signatures of $A$ and of $M(B)$ are identical, i.e.,~$\rho_A=\rho_{M(B)}$, then $C(A)=C(B)$, $U(A)=U(B)$, and $I(A)=I(B)$.
\end{theorem}
\begin{proof}
	Having already established Lemma \ref{Lem-Graphs_invariant_under_lt}, it suffices to show that if $\rho_A=\rho_B$ then $C(A)=C(B)$, $U(A)=U(B)$, and $I(A)=I(B)$. 

	Let $A_1, \dots, A_n\subset \R^2$ be translated copies of $A$ in the plane; let $v_1, \dots, v_n\in \R^2$ be vectors satisfying $A_i=A+v_i$ for every $i\in [n]$; and let $n$ translated copies of $B$, $B_1, \dots, B_n$, be defined by $B_i=B+v_i$ for $i\in [n]$. Now, consider any pair $(i, j)\in [n]^2$ and denote by $\theta_{i, j}$ and $r_{i, j}$ the argument and magnitude of the vector $v_i-v_j$. 
	By~\Cref{Lem-Moving_with_argument}, $A_i\cap A_j\neq \emptyset$ if and only if $r_{i, j}\leq\rho_A(\theta_{i, j})=\rho_B(\theta_{i, j})$, which is true if and only if $B_i\cap B_j\neq \emptyset$.  Similarly, $A_i^\circ\cap A_j^\circ\neq \emptyset$ if and only if $B_i^\circ \cap B_j^\circ\neq \emptyset$. It follows that if $A_1, \dots, A_n$ is a realization of a graph $G$ as a contact, unit distance, or intersection graph then $B_1, \dots, B_n$ is a realization of $G$ as a contact, unit distance, or intersection graph, respectively.

\end{proof}

	\section{Reducing to Symmetric Convex Bodies}\label{symmetriapp}
In this section we show that for proving the necessity of the condition of~\Cref{thm:main} it suffices to consider only symmetric convex bodies with the URTC property. 
We use the well-known trick in discrete geometry of considering the ``symmetrization'' of a convex body $A$, $K=\frac12(A+(-A))$. 
\begin{lemma}\label{finishlem}
	Let $A\subset \R^2$ be a convex body with signature $\rho_A$. Then $K=\frac12(A+(-A))$	 is a symmetric convex body with signature $\rho_K=\rho_A$.
\end{lemma}
\begin{proof}
	It is well-known and easy to check that $K$ is symmetric, bounded and convex. 	
	For the statement concerning the signature, let $\theta\in [0, \pi)$ be given. 
	Consider a profile $\ell_\theta$ through $A$ at angle $\theta$. Then $\frac12\left(\ell_\theta + (-\ell_\theta)\right)\subset K$ is a line segment of the same length and argument as $\ell_\theta$ and hence, $\rho_K(\theta)\geq \rho_A(\theta)$. 
	Conversely, consider a profile $\ell_\theta'$ through $K$ at angle $\theta$ and let $\frac12(a_1-a_2)$ and $\frac12(a_3-a_4)$ be the endpoints of $\ell_\theta'$ with $a_1, \dots, a_4\in A$. Then the vector $v = \frac12(a_1-a_2)-\frac12(a_3-a_4)$ has argument $\theta$ and magnitude $\rho_K(\theta)$, but since 
	$$\frac12(a_1-a_2)-\frac12(a_3-a_4) = \frac12(a_1+a_4)-\frac12(a_2+a_3)$$
	and $\frac12(a_1+a_4),\frac12(a_2+a_3)\in A$ by convexity of $A$, there is a line segment from $\frac12(a_1+a_4)$ to $\frac12(a_2+a_3)$ in $A$ with the same argument and magnitude as $v$. It follows that $\rho_K(\theta)\leq \rho_A(\theta)$.
\end{proof}

Thus, for every convex body $A$ there is a symmetric convex body $K$ with the same signature and which by~\Cref{suff:thm} therefore has the same contact, unit distance, and intersection graphs as $A$. 
\begin{proposition}\label{finishprop}
    Let $A$ and $B$ be convex bodies such that $\rho_A=\rho_B$. Then $A$ has the URTC property if and only if $B$ has the URTC property. In particular this applies when $B=\frac{1}{2}(A+(-A))$ is the symmetrization of $A$.
\end{proposition}
\begin{proof}
   Suppose that $A$ has the URTC property. Let $B_1=B+v_1$ and $B_2=B+v_2$ be translates of $B$ satisfying that $B_1\cap B_2\neq \emptyset$ but $B_1^\circ \cap B_2^\circ=\emptyset$. Let $A_1=A+v_1$ and $A_2=A+v_2$.  Since $\rho_A=\rho_B$, ~\Cref{Lem-Moving_with_argument} implies that $A_1\cap A_2\neq \emptyset$ but $A_1^\circ \cap A_2^\circ=\emptyset$. Again using~\Cref{Lem-Moving_with_argument} we obtain that for a vector $v_3$ it holds that $B+v_3$ intersects $B_1$ and $B_2$ but only at their boundary if and only if  $A+v_3$ intersects $A_1$ and $A_2$ but only on at the boundary. As $A$ has the URTC property it follows that there are exactly two choices of $v_3$ such that $B+v_3$ intersect $B_1$ and $B_2$ but only at their boundary. Since $B_1$ and $B_2$ were arbitrary it follows that $B$ has the URTC property. The converse implication is identical.
\end{proof}
Combined with the work done in the main body of this paper (Section~\ref{necsec} and~\ref{intsec})~\Cref{thm:main} follows immediately:
\begin{proof}[Proof of~\Cref{thm:main}]
We have already seen the sufficiency (Theorem~\ref{suff:thm}) of the condition.
Theorems~\ref{main:contact} and~\ref{theorem:main} (to be proved in the following sections) give the necessity in case $A$ and $B$ are symmetric. Suppose now that $A$ and $B$ are arbitrary convex bodies with the URTC property satisfying that for any linear transformation $B'$ of $B$, $\rho_{A}\neq \rho_{B'}$. For $X\in \{A,B\}$ let $K_X=\frac{1}{2}(X+(-X))$. If there exists a linear transformation $T$ satisfying that $\rho_{K_A}=\rho_{T(K_B)}$ then symmetri yields that $K_A=T(K_B)=K_{T(B)}$. Applying~\Cref{finishlem} we obtain the contradiction that $\rho_A=\rho_{T(B)}$ and we conclude that for no linear transformation $T$ is $\rho_{K_A}=\rho_{T(K_B)}$. Since $K_A$ and $K_B$ are symmetric and by~\Cref{finishprop} both have the URTC property, we conclude that $C(K_A)\neq C(K_B)$, $U(K_A)\neq U(K_B)$, and $I(K_A)\neq I(K_B)$.
But $C(K_X)=C(X)$, $U(K_X)=U(X)$, and $I(K_X)=I(X)$ for $X\in \{A,B\}$ by~\Cref{suff:thm}, so the result follows.
\end{proof}



	\section{Necessity for Contact and Unit Distance Graphs}\label{necsec}
\iffull
In this section we prove the necessity of the condition of Theorem \ref{thm:main} in the case of contact graphs in the setting where $A$ and $B$ are symmetric. The proof for unit distance graphs is completely identical so we will merely provide a remark justifying this claim by the end of the section. The main result of the section is slightly more general than required since we will use it in the classification of intersection graphs. 
\fi

\iffull
\subsection{Properties of the Signature}
\fi

\iffull
\iffull
Towards proving the main theorem of the section, we prove three lemmas regarding the behaviour of the signature of a symmetric convex body. The content of the concluding lemma is as follows: If $A$ and $B$, are symmetric convex bodies satisfying that no linear map transforms $A$ into $B$, then there exists a finite set of angles, $\theta_1, \dots, \theta_n$ and an $\eps>0$, such that the signature of no linear transform of $B$ is $\varepsilon$-close to the signature of $A$ at every angle $\theta_i$. This observation motivates the constructions of the section to follow.
\else
First a remark:  We defined $\rho_A(\theta)=2 \norm{e_A(\theta)}{2}$ where $e_A(\theta)$ had argument $\theta$ and $\norm{e_A(\theta)}{2}=1$. Note that this definition matches that of the profile from~\Cref{suffapp} in the case where $A$ is symmetric.

In order to prove~Lemma~\ref{findirlem}, we start out by proving two lemmata concerning continuity properties of the signature $\rho_A$ when $A$ is symmetric. 
\fi
\begin{lemma}\label{anglecont}
For a symmetric convex body $A\subset \R^2$ the signature $\rho_A\colon [0,\pi)\to \R$ is continuous.
\end{lemma}
\begin{proof}
For $\theta \in [0,\pi)$ we define the rotation matrix $M_{\theta}=
\begin{pmatrix}
    \cos \theta       & -\sin \theta  \\
    \sin \theta & \cos \theta
\end{pmatrix}$.
We further let $e_1=(1,0)\in \R^2$. It is easy to check that the mapping $\theta\mapsto M_\theta e_1$ is a continuous map $[0,\pi) \to \R^2$. Furthermore, since any norm on $\R^2$ induces the same topology as that of the Euclidian norm, the map $\norm{\cdot }{A}:\R^2\to \R_{\geq 0}$ is continuous with respect to the standard topology on $\R^2$. Thus, the map $\varphi\colon [0,\pi) \to \R^2$ defined by 
$
\varphi(\theta)=\frac{M_\theta e_1}{\|M_\theta e_1\|_A},
$
is continuous. The conclusion follows as $\rho_A(\theta)= 2\|\varphi(\theta)\|_2$.
\end{proof}

\begin{lemma}\label{operatorcont}
Let $A$ be a symmetric convex body. Let $K$ be the set of all non-singular linear maps $\R^2 \to \R^2$ with topology induced by the operator norm. For each $\theta \in [0,\pi)$ the map $f_{\theta}\colon K \to \R_{\geq 0}$ defined by $f_{\theta}(T)=\rho_{T(A)}(\theta)$ is continuous.
\end{lemma}
\begin{proof}
If $T$ is non-singular, $T(A)$ is also symmetric with non-empty interior and thus induces a norm $\norm{\cdot}{T(A)}$ on $\R^2$. Let $M_\theta$ be as in the proof of Lemma \ref{anglecont} and $v_\theta=M_\theta e_1$. Then
$$
\rho_{T(A)}(\theta)=\frac{2\|v_{\theta} \|_2}{\|v_{\theta}\|_{T(A)}}.
$$
Note that 
$$
\|v_\theta\|_{T(A)}=\inf\{\lambda\geq 0\mid  v_\theta\in \lambda T(A)\}=\inf\{\lambda\geq 0\mid  T^{-1}v_\theta \in \lambda A\}=\|T^{-1}v_\theta\|_{A},
$$
and so it suffices to show that the mapping $T\mapsto T^{-1}v_\theta$ is a continuous map $K\to \R^2$. Now, for $T_1,T_2 \in K$,
$$
\|T_1^{-1}v_\theta-T_2^{-1}v_\theta\|_2\leq \|T_1^{-1}-T_2^{-1} \|_{op} \|v_\theta\|_2,
$$
where $\|\cdot \|_{op}$ is the operator norm, so it suffices to show that the inversion $T\mapsto T^{-1}$ is a continuous map $K\to K$. It is a standard result from the literature that on the set of invertible elements of a unital Banach algebra, $\mathcal A$, the operation of inversion, $x\mapsto x^{-1}$, is continuous. Since $K$ is exactly the invertible elements of the Banach algebra of linear maps $\R^2\to\R^2$, the conclusion follows.
\end{proof}

\begin{figure}
\centering
\includegraphics{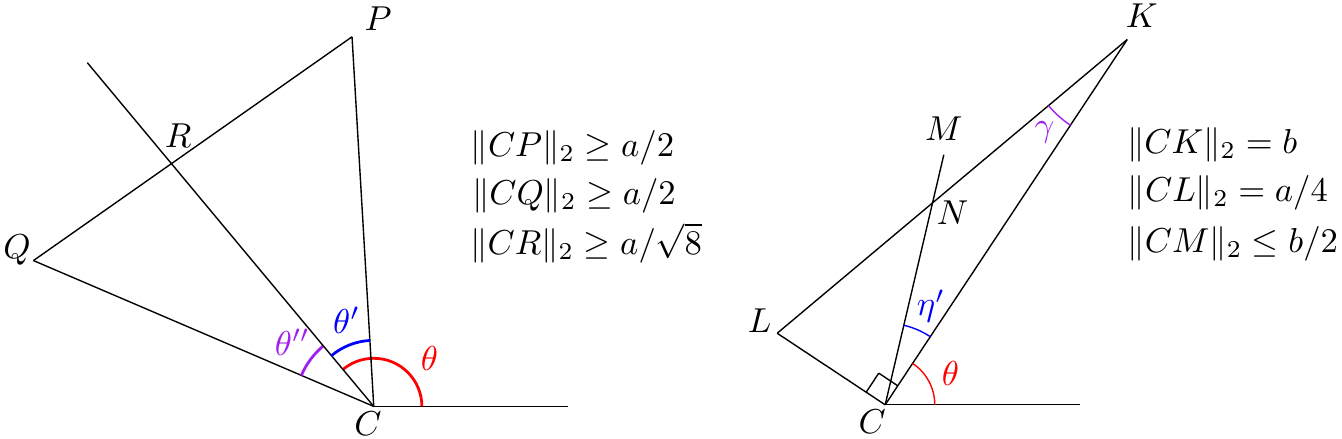}  
\caption{The two steps of the proof of Lemma~\ref{lem:bounded_signature}}
\label{geometry}
\end{figure}

\begin{lemma}\label{lem:bounded_signature}
    Let $A$ be a symmetric convex body and $a, b\in \R_{>0}$ constants. Consider some $\eta>0$ and suppose that for every $\theta\in [0, \pi)$ there exists $\theta'\in [\theta, \theta+\eta)$ (with coordinates modulo $\pi$) such that $a\leq \rho(\theta')\leq b$. If $\eta$ is sufficiently small as a function of $a$ and $b$, $\rho_A(\theta)< 2b$ for every $\theta\in [0, \pi)$.
\end{lemma}
\begin{proof}
    We compute all coordinates modulo $\pi$. Also let $C$ be the center of $A$.
    
    For the first part of the argument, see the left-hand side of Figure~\ref{geometry}. We start by letting $\eta<\pi/4$. Then for arbitrary $\theta\in [0, \pi)$ there exists $\theta', \theta''\in [0, \pi/4)$ and points $P, Q\in A$ such that the line segments $CP$ and $CQ$ have arguments $\theta-\theta'$ and $\theta+\theta''$, respectively, and satisfy $\norm{CP}{2}, \norm{CQ}{2}\geq \frac a2$. By convexity, the line segment $PQ$ is contained in $A$ and furthermore, it is easy to verify that every point on $PQ$ has distance at least $\frac{a}{\sqrt 8}$ to $C$. Since the profile of $A$ at angle $\theta$ which passes through $C$ intersects $PQ$, it follows that $\rho_A(\theta)\geq \frac{a}{\sqrt{2}}$. Thus, for all $\theta\in [0, \pi)$, $\rho_A(\theta)\geq \frac{a}{\sqrt{2}}$.

    For the remaining argument, see the right-hand side of Figure~\ref{geometry}. Towards our main conclusion, suppose for contradiction that there is a line segment $CK$ of length $b$ and argument $\theta$ contained in $A$. Let $CL$ be the line segment of argument $\theta+\pi/2$ and length $a/4$ and note that it is contained in $A$. By assumption there exists a boundary point $M$ of $A$ such that $\norm{CM}{2}\leq b/2$ and $CM$ has argument $\theta+\eta'$ where $0<\eta'<\eta$. By convexity, $CM$ intersects $LK$ at a point $N$. Denote by $\gamma$ the angle $\angle CKL$ and note that it is a constant depending only on $a$ and $b$. Now, by the law of sines in $\triangle CKN$, 
        $$
        \norm{CN}{2}=\frac{b\sin(\gamma)}{\sin(\pi-\gamma-\eta')}=b\cdot \frac{\sin(\gamma)}{\sin(\gamma+\eta')},
        $$
        so for $\eta$ sufficiently small as a function of $a$ and $b$, $\norm{CN}{2}>b/2$, which contradicts $\norm{CM}{2}\leq b/2$.
\end{proof}
\iffull
\begin{lemma}\label{findirlem}
Let $A$ and $B$ be symmetric convex bodies in $\R^2$. 
Suppose that for every finite set $\Theta \subset [0,\pi)$ and for every $\eps>0$, there exists a linear map $T\colon \R^2 \to \R^2$ satisfying that $|\rho_{T(B)}(\theta)-\rho_A(\theta)|<\eps$ for all $\theta\in \Theta$. 
Then there exists a linear map $T\colon \R^2\to \R^2$ with $T(B)=A$.
\end{lemma}
\begin{proof}
\else
We are now ready to prove~\Cref{findirlem}:

\begin{proof}[Proof of Lemma~\ref{findirlem}]
\fi
We may clearly assume that $A$ and $B$ are centered at the origin. For $n\in \N$, let $\Theta_n=\{i \pi /2^n\mid i \in [2^n]\}$; let $\eps_n=1/n$; and let $T_n:\R^2\to \R^2$ be a linear map satisfying that  $|\rho_{T_n(B)}(\theta)-\rho_A(\theta)|<\eps_n$ for all $\theta\in \Theta_n$. 

For $x\in \R^2$ and $r>0$ denote by $\ball{x}{r}=\{y\in \R^2\mid \|x-y\|_2<r\}$ the open ball in the Euclidian norm with center $x$ and radius $r$. 
We begin by proving that $\{T_n\mid n\in \N\}$ is uniformly bounded in the operator norm by showing that there exists a constant $R>0$ such that $T_n(B)\subset\ball{0}{R}$ when $n$ is sufficiently large. To this end, let $a=\inf_{\theta\in [0, \pi)}\rho_A(\theta)$ and $b=\sup_{\theta\in [0, \pi)}\rho_A(\theta)$ and note that $0<a, b <\infty$. There exists $N_0\in \N$ and constants $a', b'>0$ such that for every $n\geq N_0$, $a-\eps_n\geq a'$ and $b+\eps_n\leq b'$. This implies that for every $n\geq N_0$ and $\theta\in \Theta_n$, $a'\leq \rho_{T_n(B)}(\theta)\leq b'$. Applying Lemma \ref{lem:bounded_signature}, we find that there exists an $N_1\in \N$ such that for every $n\geq N_1$, $\rho_{T_n(B)}<2b'$. Thus, $T_n(B)\subset\ball{0}{b'}$ for $n\geq N_1$ so $\{T_n\mid n\in \N\}$ is uniformly bounded. As it moreover holds for $n\geq N_1$ and $\theta \in \Theta_n(\theta)$ that $\rho_{T_n(B)}\geq a'$, convexity of $T_n(B)$ gives that $\rho_{T_n(B)}\geq \frac{a'}{\sqrt{2}}$ for $n\geq N_1$. In particular $T_n(B)\supset \ball{0}{a'/\sqrt{8}}$.

Since $(T_n)_{n>0}$ is uniformly bounded we may by compactness assume that $(T_n)_{n>0}$ converges in operator norm to some linear map $T:\R^2 \to \R^2$ by passing to an appropriate subsequence. As moreover $T_n(B)\supset \ball{0}{a'/\sqrt{8}}$, it is easy to check that $T(B)\supset \ball{0}{a'/\sqrt{8}}$ so in particular $T$ is non-singular. We claim that $T(B)=A$. As $A$ and $T(B)$ are symmetric it suffices to show that $\rho_A=\rho_{T(B)}$. Moreover, $\rho_A$ and $\rho_{T(B)}$ are both continuous by Lemma~\ref{anglecont} so since $\bigcup_{n\in \N} \Theta_n$ is dense in $[0,\pi)$ it suffices to show that $\rho_A|_{\Theta_n}=\rho_{T(B)}|_{\Theta_n}$ for each $n\in \N$. 
To see this let $\theta\in \bigcup_{n\in \N} \Theta_n$ and let $f_\theta$ be defined as in Lemma~\ref{operatorcont}. Then $f_\theta$ is continuous so $f_\theta(T_n)\to f_\theta(T)$ as $n\to \infty$ (here we use that $T$ is non-singular). It follows that
$$
\rho_{T(B)}(\theta)=\lim_{n\to \infty}\rho_{T_n(B)}(\theta)=\rho_A(\theta),
$$
as desired. This completes the proof.
\end{proof}

\else
Due to space limitations we have left out the proof, but it can be found in the full version~\cite{}.
\fi
 \iffull
\subsection{Establishing Necessity}
 Before proving the part of Theorem~\ref{thm:main} concerning contact graphs we describe certain lattices which gives rise to contact graphs that can be realised in an essentially unique way. We start with the following definition.
 \fi
\begin{definition}
Let $A\subset \R^2$ be a symmetric convex body with the URTC property, and $\norm{\cdot}{A}$ the associated norm. Let $e_1,e_2\in \R^2$ be such that $\norm{e_1}{A}=\norm{e_2}{A}=\norm{e_1-e_2}{A}=2$. We define the lattice $\mathcal{L}_A(e_1,e_2)=\{a_1e_1+a_2e_2\mid (a_1,a_2)\in \Z^2\}$. 
\end{definition}
Note that if $e_1$ has been chosen with $\|e_1\|_A=2$, then using the URTC property there are  precisely two vectors $v$ with $\|v\|_A=\|v-e_1\|_A=2$. If one is $v_2$ the second is $e_1-v_2$ so regardless how we choose $e_2$ we obtain the same lattice. 
Let us describe a few properties of the lattice $\mathcal{L}_A(e_1,e_2)$.
Using the triangle inequality and the URTC property of $A$ it is easily verified that for distinct $x,y\in \mathcal{L}_A(e_1,e_2)$, $\|x-y\|_A\geq 2$ with equality holding exactly if $x-y\in \mathcal{S}_A:=\{e_1,e_2,-e_1,-e_2,e_1-e_2,e_2-e_1\}$.
Another useful fact is the following:

\begin{lemma}\label{normrellem}
With $\mathcal{S}_A$ as above it holds that $\frac{1}{2}\conv(\mathcal{S}_A)\subset A\subset \conv(\mathcal{S}_A)$. Here $\conv(\mathcal{S}_A)$ is the convex hull of $\mathcal{S}_A$. If in particular $B$ is another symmetric convex body for which $\norm{e_1}{B}=\norm{e_2}{B}=\norm{e_1-e_2}{B}=2$, then for all $x\in \R^2$ it holds that $\frac{1}{2}\norm{x}{A}\leq \norm{x}{B}\leq 2\norm{x}{A}$.
\end{lemma}
\begin{proof}
As $\frac{1}{2}\mathcal{S}_A\subset A$ and $A$ is convex the first inclusion is clear. For the second inclusion we note that all points $y$ on the hexagon connecting the points $e_1,e_2,e_2-e_1,-e_1,-e_2,e_1-e_2$ of $\mathcal{S}_A$ in this order has $\|y\|_A\geq 1$ by the triangle inequality and so $A\subset\conv(\mathcal{S}_A)$.

For the last statement of the lemma note that if $x\in \R^2$ then
\begin{align*}
\norm{x}{B} & =\inf_{\lambda\geq 0}\{x \in \lambda B\}\geq \inf _{\lambda\geq 0}\{ x \in \lambda \conv(\mathcal{S}_B)\}\\
&=\inf _{\lambda\geq 0}\{ x \in \lambda\conv(\mathcal{S}_A)\}\geq \inf _{\lambda\geq 0}\{x \in 2\lambda A\}=\frac{1}{2}\norm{x}{A},
\end{align*}
and similarly $\|x\|_A\geq  \frac 12 \|x\|_B$.
\end{proof}

\begin{definition}\label{def:latticeu}
We say that a graph $G=(V,E)$ is \emph{lattice unique} if $|V|=n\geq 3$ and there exists an enumeration of its vertices $v_1,\dots,v_n$ such that
\begin{itemize}
\item The  vertex induced subgraph $G[v_1,v_2,v_3]\simeq K_3$ is a triangle.
\item For $i>3$ there exists distinct $j,k,l<i$ such that $G[v_j,v_k, v_l]\simeq K_3$ and both $(v_i,v_j)$ and $(v_i,v_k)$ are edges of $G$.
\end{itemize}
\end{definition}

Suppose that $A$ is a symmetric convex body with the URTC property, that $\mathcal{A}\subset \R^2$ is compatible with $A$, and that $G=C_A(\mathcal{A})$ is lattice unique. Enumerate the points of $\mathcal{A}=\{v_1,\dots,v_n\}$ according to the definition of lattice uniqueness. Without loss of generality assume that $v_1=0$. 
Then the URTC property of $A$ combined with the lattice uniqueness of $G$ gives that $v_4,\dots,v_n$ are uniquely determined from $v_2$ and $v_3$ and all contained in $\mathcal{L}_A(v_2,v_3)$. If moreover $B$ is another convex body with the URTC property, $\mathcal{B}=\{v_1',\dots,v_n'\}\subset \R^2$  has $v_1'=0$ and is compatible with $B$, and $C_B(\mathcal{B})\simeq C_A(\mathcal{A})$ via the graph isomorphism $\varphi:v_i' \mapsto v_i$, then the linear map $T:\R^2 \to \R^2$ defined by $T:a_1v_2'+a_2v_3'\mapsto a_1v_2+a_2v_3$ satisfies that $T|_{\mathcal{B}}=\varphi$. 
\\

Before commencing the proof of Theorem~\ref{main:contact} let us highlight the main ideas. The most important tool is Lemma~\ref{findirlem} according to which there exist $\eps>0$ and a finite set of directions $\Theta$ such that for any linear tranformation $B'$ of $B$ there is a direction $\theta\in\Theta$ such that $\rho_A(\theta)$ and $\rho_{B'}(\theta)$ differ by at least $\eps$.
We will construct $G$ by describing a finite set $\mathcal{A}\subset \R^2$ compatible with $A$, and defining $G=C_A(\mathcal{A})$. Now, $\mathcal{A}$ will be a disjoint union of two sets of points, $\mathcal{A}=\mathcal{U}\cup \mathcal{W}$, where $\mathcal{U}$ and $\mathcal{W}$ will play complementary roles. 
The construction will be such that $\mathcal{U}$ is a subset of a lattice $\mathcal{L}=\mathcal{L}_A(e_1,e_2)$ and such that the corresponding induced subgraph $G[\mathcal{U}]$ of $G$ is lattice unique. More precisely $\mathcal{U}$ will consist of $|\Theta|$ large hexagons connected along their edges.
When attempting to realize $G$ as a contact graph of $B$ the lattice uniqueness enforces that $G[\mathcal{U}]$ is realized as a subgraph of a lattice $\mathcal{L}_B(e_1',e_2')$ in essentially the same way. 
The remaining points of $\mathcal{W}$ do not lie in the lattice $\mathcal{L}$. 
They constitute \emph{rigid beams} in the directions from $\Theta$ ``connecting'' diagonally opposite points of the $|\Theta|$ hexagons of $\mathcal{U}$. The construction of $\mathcal{A}$ is depicted in the left-hand side of Figure~\ref{hexagon} and in Figure~\ref{hexagon3}.
When trying to reconstruct the same contact graph (or a supergraph) with beams connecting the corresponding points of $\mathcal{U}'$, we will find that in at least one direction the beam becomes too long or too short.

\begin{proof}[Proof of Theorem~\ref{main:contact}]

We let $e_1,e_2\in \R^2$ be such that $\|e_1\|_A=\|e_2\|_A=\|e_1-e_2\|_A=2$ and
define the lattice $\mathcal{L}\mydef\mathcal{L}_A(e_1,e_2)$. We also define
the infinite graph $G_0 \mydef C_A(\mathcal{L})$. Without loss of generality
we can assume that $e_1$ and $e_2$ satisfy
that $\|e_1\|_2=\|e_2\|_2=\|e_1-e_2\|_2=2$, since there exists a non-singular
linear transformation $T$ such that $\|T(e_1)\|_2=\|T(e_2)\|_2=\|T(e_1)-T(e_2)\|_2=2$,
and $C(A) = C(T(A))$. Note that in this setting we can use Lemma \ref{normrellem} to compare $A$
to the circle of radius 1 and obtain $\frac12 \| x\|_2\leq \|x\|_A\leq 2\|x\|_2$ for every $x\in \R^2$.



\begin{figure}
\centering
\includegraphics{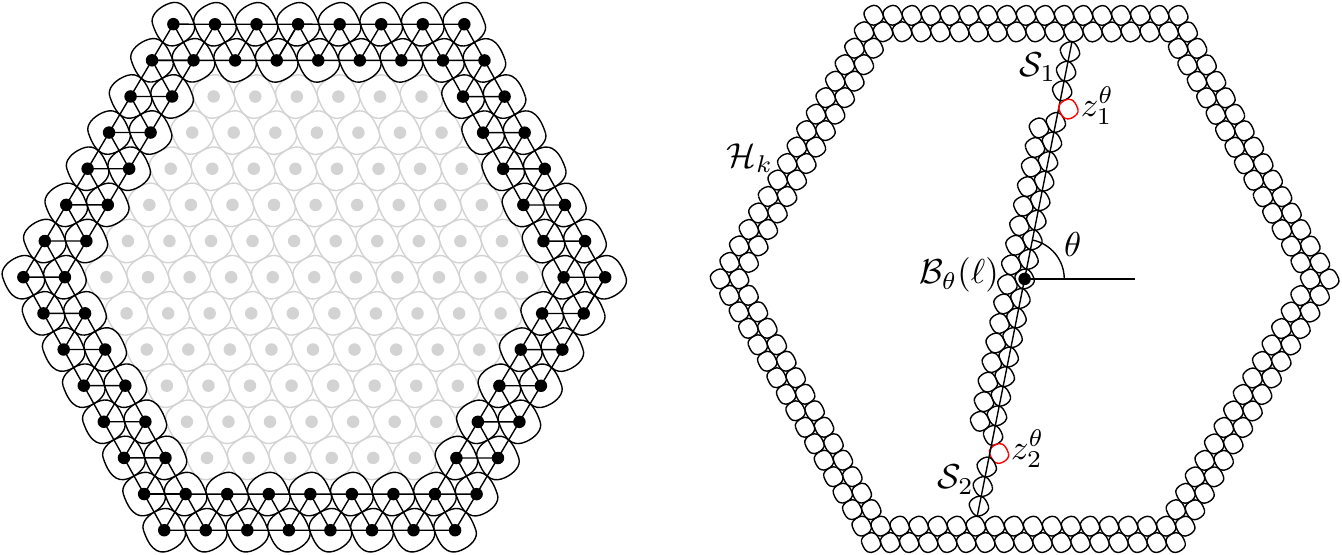}
\caption{Left: The points of $\mathcal{H}_6$ along with the corresponding lattice unique subgraph $G_0[\mathcal{H}_6]$. Right: The attachment of the beam $\mathcal{B}_{\theta}(\ell)$.}
\label{hexagon}
\end{figure}

As already mentioned we will construct $G$ by specifying a finite point set $\mathcal{A}\subset \R^2$ compatible with $A$ and define $G=C_A(\mathcal{A})$. The construction of $\mathcal{A}$ can be divided into several sub-constructions. We start out by describing a hexagon of points $\mathcal{H}_k$ for $k\in \N$ which satisfies that $C_A(\mathcal{H}_k)$ is lattice unique.

\begin{construction}[$\mathcal{H}_k$]
For an illustration of the construction see the left-hand side of Figure~\ref{hexagon}.
For $x,y\in \mathcal{L}$ we write $d(x,y)$ for the distance between $x$ and $y$ in the graph $G_0$, and for $k\in \N$ we define $\mathcal{H}_k=\{x\in \mathcal{L}\mid  d(x,0)\in \{k,k+1\}\}$. Clearly $G_0[\mathcal{H}_k]$ is a lattice unique graph. Moreover, using that $e_1$ and $e_2$ satisfy $\|e_1\|_2=\|e_2\|_2=\|e_1-e_2\|_2=2$ it is easy to check that the points $\{x\in \mathcal{L}\mid  d(x,0)=k\}\subset \mathcal{H}_k$ lie on a regular hexagon $H_k$ whose corners have distance exactly $2k$ to the origin in the Euclidian norm. In particular the points $p\in \mathcal{H}_k$ has $\|p\|_2\geq \sqrt{3}k$, and thus $\|p\|_A\geq \frac{\sqrt{3}}{2}k$ by Lemma~\ref{normrellem}.
\end{construction}

For a given $\theta\in [0,\pi)$ and $\ell\in \N$ we will construct a set of points $\mathcal{B}_\theta(\ell)\subset \R^2$ compatible with $A$ which constitute a ``beam'' of argument $\theta$:
\begin{construction}[$\mathcal{B}_\theta{(\ell)}$]
Let $e_{\theta}\in \R^2$ be the vector of argument $\theta$ with $\|e_{\theta}\|_A=2$, and let $f_{\theta}\in \R^2$ be such that $\|f_{\theta}\|_A=\|f_{\theta}-e_{\theta}\|_A=2$ (by the URTC property we have two choices for $f_{\theta}$). For a given $\ell\in \N$ we define 
\[
    \mathcal{B}_\theta(\ell)=\{a e_\theta\mid  a \in \{-\ell,\dots,\ell\}\}\cup \{ae_\theta+f_{\theta}\mid  a \in \{-\ell,\dots, \ell-1\}\}
\]
Note that $\mathcal{B}_\theta(\ell)$ is compatible with $A$ and that $C(\mathcal{B}_\theta(\ell))$ is lattice unique. 
\end{construction}

For a given $k$ we want to choose $\ell$ as large as possible such that $\mathcal{B}_\theta(\ell)$ ``fits inside'' $G_0[\mathcal{H}_k]$. We then wish to ``attach'' $\mathcal{B}_\theta(\ell)$ to $G_0[\mathcal{H}_k]$ with extra points $\mathcal{S}$, the number of which does neither depend on $k$ nor on $\theta$. We wish to do it in such a way that $\mathcal{A}_1^k(\theta):= \mathcal{B}_\theta(\ell)\cup G_0[\mathcal{H}_k] \cup \mathcal{S}$ is compatible with $A$. The precise construction is as follows:

\begin{figure}
\centering
\includegraphics{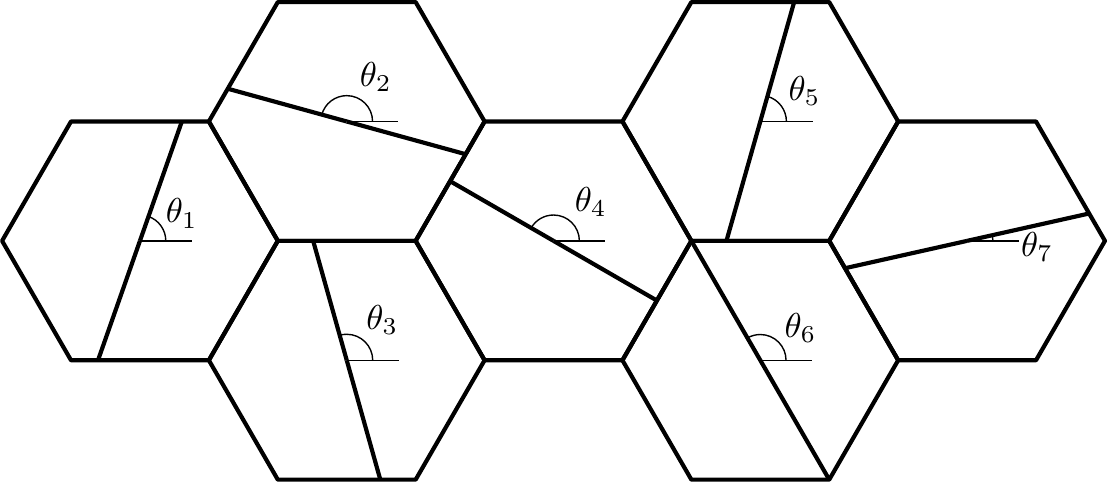}
\caption{The final point set $\mathcal{A}$ where the point sets $\mathcal C_k(\theta)$ are ``glued'' together by translating them such that the contact graph realized by the union of the subsets $\mathcal H_k\subset \mathcal C_k(\theta)$ is lattice unique.}
\label{hexagon3}
\end{figure}

\begin{construction}[$\mathcal{C}_k(\theta)$]
See Figure~\ref{hexagon} (right).
Consider the open line segment $L_{\theta}=\{re_{\theta}\mid r \in (-r_{\max},r_{\max})\}$ where $r_{\max}$ is maximal with the property that for all points $x\in L_{\theta}$ and all $y\in H_k$ it holds that $\|x-y\|_A>4$. Also let $\ell\in \N$ be maximal such that $\{a e_\theta\mid  a \in \{-\ell,\dots,\ell\}\}\subset L_{\theta}$. Note that $\ell \geq \frac{\sqrt{3}}{4}k-3$ as the points $p\in H_k$ has $\|p\|_A\geq \frac{\sqrt{3}}{2}k$. If in particular $k> \frac{12}{\sqrt{3}-1}$ it holds that $\ell > \frac{k}{4}$.

When $\ell$ is chosen in this fashion, we have that $\mathcal{B}_\theta(\ell)$ is contained in the interior of $H_k$. Further, for all points $x\in \mathcal{B}_\theta(\ell)$ and all $y\in H_k$, it holds by the triangle inequality that $\|x-y\|_A>2$ since by construction every point of $\mathcal B_\theta(\ell)$ has distance at most 2 to $L_\theta$ in the norm $\norm{\cdot}{A}$. 
Now,  $\mathcal{B}_\theta(\ell)$ will constitute our beam in direction $\theta$ and we will proceed to show that we can attach it to $\mathcal{H}_k$, as illustrated, using only a constant number of extra points. That this can be done is conceptually unsurprising but requires a somewhat technical proof.

For this we let $s>0$ be minimal with the property that there exists and $y\in \mathcal{H}_k$ such that $\|se_{\theta}-y\|_A=2$. As $d_A(L_{\theta},\mathcal{H}_k)> 4$ it holds that $s> \ell+1$. Define $\mathcal{S}_1(\theta)=\{se_{\theta},(s-1)e_{\theta},\dots,(s-s')e_{\theta}\}$ where $s'$ is chosen maximal such that  $\mathcal{S}_1\cup \mathcal{B}_\theta(\ell)$ is compatible with $A$. In other words $s-s'-1<1+\ell \leq s-s'$. 
As $2\leq \|(s-s')e_{\theta}-\ell e_{\theta}\|_A< 4$ we can find at least one point $z_1^\theta\in \R^2$ such that $\|z_1^\theta-(s-s')e_{\theta}\|_A=\|z_1^\theta-\ell e_{\theta}\|_A=2$ (this is one of the red copies of $A$ in Figure~\ref{hexagon}). It is easy to check that $\mathcal{B}_\theta(\ell)\cup \mathcal{S}_1(\theta)\cup \left\{z_1^\theta\right\}$ is compatible with $A$. To see that $\left\{z_1^\theta\right\}\cup \mathcal{H}_k$ is also compatible with $A$ we note that for any point $x\in \mathcal{H}_k$, 
$$
\|x-z_1^\theta\|_A\geq \| x-\ell e_{\theta}\|_A-\| \ell e_{\theta}-z_1^\theta\|_A\geq 4-2=2.
$$
Finally, we need to argue that $\abs{\mathcal{S}_1(\theta)}$ is bounded by a constant independent of $\theta$ and $k$. To this end let $P=\ell e_\theta$, $Q$ a point on $H_k$ of minimal Euclidian distance to $P$, and $R$ the intersection between $H_k$ and the line $\{re_\theta \mid r \in \R\}$. It is easy to check that the points $Q$ and $R$ lie on the same edge of $H_k$ and that the angle $\angle QPR\leq \pi/6$. It follows that $\norm{PR}{2}= \norm{PQ}{2}/\cos(\angle QPR)\leq \frac{2}{\sqrt{3}}\norm{PQ}{2}$. Combining this with the fact that $\norm{PQ}{2}\leq 2\norm{PQ}{A}\leq2d_A(\{\ell e_{\theta}\},H_k)\leq 12$ we obtain

$$
\|se_{\theta}-\ell e_\theta\|_A\leq \norm{PR}{A}\leq 2\norm{PR}{2}\leq 16\sqrt{3}<28,
$$
and so $\mathcal{S}_1(\theta)$ consists of at most $13$ points.

We may similarly define $\mathcal{S}_2(\theta)$ and $z_2^\theta$ to attach the other end of the beam, $\mathcal{B}_\theta(\ell)$. Letting $\mathcal{C}_k(\theta)=\mathcal{H}_k\cup \mathcal{B}_\theta(\ell) \cup \mathcal{S}_1(\theta) \cup \mathcal{S}_2 \cup \left\{z_1^\theta,z_2^\theta\right\}$ be the combination of the components completes the construction.
\end{construction}

\begin{figure}
\centering
\includegraphics{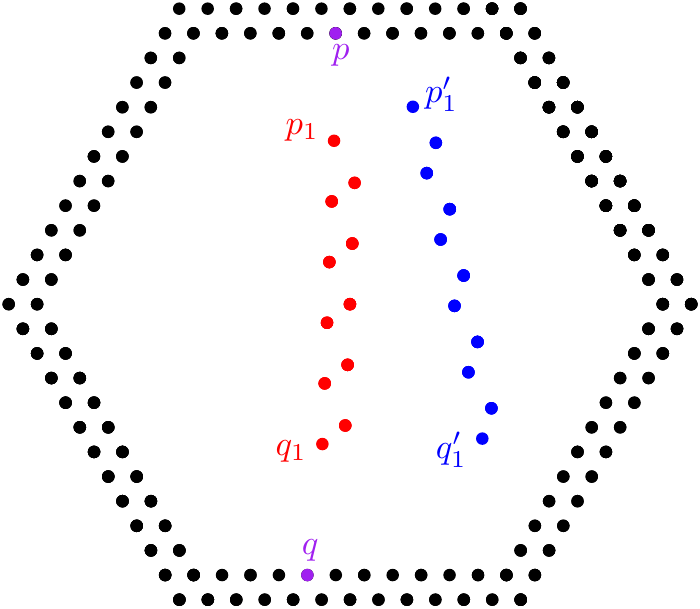}
\caption{Situation from the argument that no graph in $C(B)$ has a subgraph isomorphic to $G$.}
\label{hexagon4}
\end{figure}

We are now ready to construct $\mathcal{A}$ which will consist of several translated copies $\mathcal{C}_k(\theta)$.
\begin{construction}[$\mathcal{A}$]\label{construction:A}
By Lemma~\ref{findirlem} we can find an $\eps\in (0,1)$ and a finite set of directions $\Theta\subset [0,\pi)$ such that for all linear maps $T:\R^2 \to \R^2$ there exists $\theta \in \Theta$ such that 
$$
\left| \frac{\rho_{A}(\theta)}{\rho_{T(B)}(\theta)}-1 \right|\geq \eps.
$$
That we can scale the deviation to be multiplicative rather than additive is possible because $0<\inf_{\theta\in [0, \pi)}\rho_A(\theta)\leq \sup_{\theta\in [0, \pi)}\rho_A(\theta)<\infty$.

For each $\theta\in \Theta$ we construct a copy of $\mathcal{C}_k(\theta)=\mathcal{H}_k\cup \mathcal{B}_\theta(\ell) \cup \mathcal{S}_1(\theta) \cup \mathcal{S}_2(\theta) \cup \{z_1^\theta,z_2^\theta\}$. We then choose translations $t_{\theta}\in \R^2$ for each $\theta \in \Theta$ such that $\bigcup_{\theta \in \Theta}( \mathcal{H}_k+t_{\theta})\subset\R^2$ is compatible with $A$ and induces a lattice unique contact graph. We can choose $(t_\theta)_{\theta\in \Theta}$ in numerous ways to satisfy this. One is depicted in Figure~\ref{hexagon3}. Another is obtained by enumerating $\Theta=\{\theta_1,\dots,\theta_q\}$ and defining $t_{\theta_i}=((2k+3)e_1-(k+1)e_2) \times (i-1)$. The exact choice is not important and picking one, we define $\mathcal{A}(k)=\bigcup_{\theta \in \Theta} (\mathcal{C}_k(\theta)+t_{\theta})$ which is a point set compatible with $A$. Lastly, we set $\mathcal{A}=\mathcal{A}\left(\left\lceil\frac{180}{\eps}\right\rceil\right)$.
\end{construction}

We are now ready for the final step of the proof:
\paragraph*{Proving that no graph in $C(B)$ contains a subgraph isomorphic to $G=C_A(\mathcal{A})$.} Suppose for contradiction that there exists a set of points $\mathcal{B}\subset \R^2$ such that $G$ is isomorphic to a subgraph of $C_B(\mathcal{B})$. We may clearly assume that $|\mathcal{A}|=|\mathcal{B}|$ and we let $\varphi:\mathcal{A}\to \mathcal{B}$ be a bijection which is also a graph homomorphism when considered as a map $C_A(\mathcal{A})\to C_B(\mathcal{B})$. 
The points $\bigcup_{\theta \in \Theta}( \mathcal{H}_k+t_{\theta})$ induce a lattice unique contact graph of $A$. Thus, we may write $\bigcup_{\theta \in \Theta}( \mathcal{H}_k+t_{\theta})=\{p_1,\dots,p_n\}$ such that $p_1,p_2$ and $p_3$ induce a triangle of $G$ and such that for $i>3$ there exist distinct $j,k,l< i$ such that $p_j,p_k$ and $p_l$ induce a triangle and such that $(p_i,p_k)$ and $(p_i,p_l)$ are edges of $G$. By translating the point sets $\mathcal{A}$ and $\mathcal{B}$ we may assume that $\varphi(p_1)=p_1=0$. Then applying an appropriate linear transformation $T$, thus replacing $B$ by $T(B)$, we may assume that $\varphi(p_2)=p_2$ and $\varphi(p_3)=p_3$. Finally, the discussion succeeding~\Cref{def:latticeu} implies that in fact $\varphi|_{\bigcup_{\theta \in \Theta}( \mathcal{H}_k+t_{\theta})}$ is the identity. 

As noted in Construction \ref{construction:A}, there exists $\theta\in \Theta$ such that $\left| \frac{\rho_{A}(\theta)}{\rho_{T(B)}(\theta)}-1 \right|\geq \eps$. The outline of the remaining argument is as follows: The Euclidian length of the beam $\mathcal{B}_\theta(\ell)$ is $2\ell \rho_{A}(\theta)$, but we will see that rigidity of $\bigcup_{\theta \in \Theta}( \mathcal{H}_k+t_{\theta})$ means that it is also $ 2\ell \rho_{T(B)}(\theta)+O(1)$. When $k$ (and hence $\ell$) is large enough, this will contradict the inequality above.

Formally, assume $t_\theta=0$ without loss of generality. Let $p,q\in \mathcal{H}_k$ be such that for some $x\in \mathcal{S}_1(\theta)$, $\|p-x\|_A=2$, and for some $y\in \mathcal{S}_2(\theta)$, $\|q-y\|_A=2$. Note that $\varphi(p)=p$ and $\varphi(q)=q$. Also define $p_1=\ell e_\theta$, $q_1=-\ell e_\theta, p_1'=\varphi(p_1)$ and $q_1'=\varphi(q_1)$ (see Figure~\ref{hexagon4}). Then
\begin{align*}
&\left| \|p_1q_1\|_{T(B)} -\|p_1q_1\|_A \right|=\left|\|p_1q_1\|_{T(B)}-\|p_1'q_1'\|_{T(B)}\right| \\
\leq &\left|\|p_1q_1\|_{T(B)}-\|p_1'q_1\|_{T(B)}\right|+\left|\|p_1'q_1\|_{T(B)}-\|p_1'q_1'\|_{T(B)}\right| 
\leq \|p_1p_1'\|_{T(B)}+ \|q_1q_1'\|_{T(B)}.
\end{align*}
Next, there is a path of length $\abs{\mathcal S_1(\theta)}+2$ from $p$ to $p_1$ in $\mathcal C_k(\theta)$ with intermediate vertices $\left\{z_1(\theta)\right\}\cup \mathcal S_1(\theta)$. Combining this with Lemma~\ref{normrellem} and the fact that $\varphi(p_1)=p_1'$, we find
$$
\|p_1-p_1'\|_{T(B)}\leq \|p_1-p\|_{T(B)}+\|p-p_1'\|_{T(B)}\leq 2\|p_1-p\|_A + \|p_1-p\|_A\leq 6(|\mathcal{S}_1|+2)\leq 90.
$$ 
Similarly, $\|q_1q_1'\|_{T(B)}\leq 90$, so  $\left| \|p_1q_1\|_{T(B)} -\|p_1q_1\|_{A} \right|\leq 180$. But on the other hand we have that $\ell > k/4$, and so arrive at the contradiction
$$
\left| \|p_1q_1\|_{T(B)} -\|p_1q_1\|_{A} \right|=4 \ell \left| \frac{\rho_{A}(\theta)}{\rho_{T(B)}(\theta)}-1 \right|\geq 4  \ell\eps>k\eps=\left\lceil\frac{180}{\eps}\right\rceil \cdot  \eps\geq 180.\qedhere
$$
\end{proof}
\begin{remark}
    We claimed that the proof of the part of Theorem~\ref{thm:main} concerning unit distance graphs is identical to the proof above. 
    In fact, if we replace $C(X)$ by $U(X)$ for $X\in\{A,B\}$ in the statement of \Cref{main:contact}, the result remains valid.
    To prove it we would construct $\mathcal{A}$ in precisely the same manner.
    The important point is then that the comments immediately prior to \Cref{thm:main} concerning the rigidity of the realization of lattice unique graphs remains valid.
    If in particular $\mathcal{B}\subset \R^2$ satisfies that $U_A(\mathcal{A})\simeq U_A(B)$ via the isomorphism $\varphi:\mathcal{A}\to \mathcal{B}$, we may assume that $\varphi|_{\bigcup_{\theta \in \Theta}( \mathcal{H}_k+t_{\theta})}$ is the identity as in the proof above.
    The remaining part of the argument comparing the lengths of the beams then carries through unchanged.
    In conclusion, we are only left with the task of proving \Cref{thm:main} for intersection graphs.
\end{remark}

	\section{Intersection Graphs}\label{intsec}

In this section we prove~\Cref{theorem:main}. 
Our proof strategy is as follows:
Consider two convex bodies $A$ and $B$ as in the statement of the theorem.
We construct an intersection graph $Q^k_A \in I(A)$ containing a cycle $\alpha_k$ such that in any drawing
of $Q^k_A$ as an intersection graph, $\alpha_k$ is contained in a translation of the annulus $kA \setminus (k - 1)A$. This allows us
to view $\alpha_k$ as an upscaled copy of the boundary of $A$ with a precision error decreasing in $k$.
Similarly, in any drawing of $Q^k_A$ as an intersection graph of $B$, the cycle $\alpha_k$ is an upscaled copy of the boundary of $B$.
The idea is then to build contact graphs using $\alpha_k$ from distinct copies of $Q^k_A$.
Since we know that $C(A)\neq C(B)$, it follows that $I(A)\neq I(B)$.

However, $\alpha_k$ is not a completely fixed figure since there are many drawings of $Q^k_A$ as an intersection graph. To capture this uncertainty, we introduce the concept of \emph{$\eps$-overlap graphs}.

\begin{definition}[$\eps$-overlap Graph]
	Let $\eps>0$ and $K \subset \R^2$ be a symmetric convex body, and let 
	$v_0, \ldots, v_{n - 1} \subset \R^2$ be $n$ points in the plane. 
	Suppose that for any $i, j \in [n]$, $\norm{v_iv_j}{K}\geq 2-\eps$.
	A graph $G$ with vertex set $[n]$ and edge set satisfying
	\begin{align*}
		E(G) \subseteq \setbuilder{(i, j)\in [n]^2}{\norm{v_iv_j}{K}\leq 2}
	\end{align*}
	is called an \emph{$\eps$-overlap graph} of $K$.
	We say that $\{v_0, \ldots, v_{n - 1}\}$ \emph{realize} the graph $G$ as an $\eps$-overlap graph of $K$.
	Further, we denote by $C_\eps(K)$ the set of graphs that can be realized as $\eps$-overlap graphs of $K$.
\end{definition}

We will use $\alpha_k$ to build an $\eps$-overlap graph where 
$\eps = O\left(\frac{1}{k}\right)$. We place copies of $Q^k_A$ centered at every point
$v_0, \ldots, v_{n - 1}\in\R^2$, which are the vertices of the graph, and say that there is an edge
between two points $v$, $v'$ if the corresponding cycles $\alpha_k$, $\alpha_k'$
intersect. Then using the following reduction from $\eps$-overlap graphs to contact
graphs finishes our proof.

\begin{lemma}\label{lemma:implyContact}
	Consider a graph $G'=(V,E')$ with $V=[n]$ and a convex body $A$.
	If for every $\eps>0$, $G'\in C_\eps(A)$, then there is a graph $G=(V,E)\in
	C(A)$ such that $E'\subseteq E$.
\end{lemma}

\begin{proof}
	Let $\eps_m$ be a positive sequence such that $\eps_m\longrightarrow 0$ as
	$m\longrightarrow \infty$ and suppose that $\mathcal
	A^m=\{v^m_1,\ldots,v^m_n\}\subset \R^2$ realize $G'$ as an $\eps_m$-overlap graph.
	We may clearly assume that the values $\|v^m_i\|_A$ are bounded.
	By passing from $\mathcal A^m$ to a subsequence, we may therefore assume that
	for each $i\in V$, $v^m_i$ converges to some point $v_i$ as
	$m\longrightarrow\infty$.
	Clearly, $\|v_iv_j\|_A\geq 2$, so $\{v_1,\ldots,v_n\}$ are compatible with $A$ and define a contact graph
	$G=(V,E)=C_A(\{v_1,\ldots,v_n\})$.
	Furthermore, if $ij\in E'$, then $\|v_iv_j\|_A\leq 2$, so $E'\subseteq E$.
\end{proof}

Combining this with \Cref{main:contact} will exactly give us our result. So the
rest of this section will be dedicated to showing that the graph $Q^k_A$ exists
and describe how to build $\eps$-overlap graphs using it.





In the construction of $Q^k_A$ we have a designated vertex $s_0$ with the
property that for every drawing of $Q^k_A$ as an intersection graph and every vertex $v\in \alpha_k$, we have $\norm{s_0v}{A}=\Omega(k)$.
To obtain this property, we first construct another graph $P^k_A$ (which will be contained in $Q^k_A$)
with a vertex $s_0$ such that in every drawing of $P^k_A$
as an intersection graph, $s_0$ is contained in $k$ nested disjoint
cycles. A priori, it is not clear what it means for $s_0$ to be contained in a cycle
of the graph in every drawing, since the drawing is not necessarily a plane
embedding of the graph.
However, as the following lemma shows, it is well-defined if $P^k_A$ is triangle-free.

\begin{lemma}\label{lemma:triangle}
	If $G$ is a triangle-free graph then every drawing of $G$ as an intersection
	graph is a plane embedding.
\end{lemma}
\begin{proof}
	The proof will be by contraposition so assume that $A \subset \R^2$ is a
	convex body and that $G=(V,E)$ is a drawing as an intersection graph. Then there exists
	$x, y, z, w \in V$ with $xy, zw \in E$ and where the edges $xy$
	and $zw$ intersect. Call this intersection point $p \in \R^2$. We know
	that $\norm{xy}{A} \le 2$ and $\norm{zw}{A} \le 2$. Using the triangle
	inequality we get that
	\begin{align*}
		\norm{xw}{A} &\le \norm{xp}{A} + \norm{wp}{A} \\
		\norm{yz}{A} &\le \norm{yp}{A} + \norm{zp}{A}
	\end{align*}
	combining this we get that
	\[
		\norm{xw}{A} + \norm{yz}{A} \le \norm{xy}{A} + \norm{zw}{A} \le 4
	\]
	since $p$ lies on the lines $xy$ and $zw$. This implies that either
	$\norm{xw}{A} \le 2$ or $\norm{yz}{A} \le 2$ so either $xw \in E(G)$
	or $yz \in E(G)$. An analogous argument shows that either $xz \in E(G)$ or
	$yw \in E(G)$. This shows that $G$ contains a triangle which finishes the
	proof.
\end{proof}

We are now ready to define $P^k_A \in I(A)$ for any $k > 0$.
Besides being triangle-free, our aim is that $P^k_A$ should have the following properties:
\begin{enumerate}
\item\label{prop:P1}
There is a vertex $s_0$ such that in any drawing of $P_A^k$ as an intersection graph of $A$ and $B$, $s_0$ is contained in $k$ nested disjoint, simple cycles $\sigma_1,\ldots,\sigma_k$.

\item\label{prop:P2}
There is a path $\kappa_k$ from a vertex $s_k$ to a leaf $t_k$ such that in any drawing of $P_A^k$ as an intersection graph of $A$ and $B$, the path $\kappa_k$ is on the boundary of the outer face.
%
\end{enumerate}

\begin{construction}[$P_A^k$.]
As in the previous section choose $e_1,e_2\in \R^2$ such that $\norm{e_1}{A}=\norm{e_2}{A}=\norm{e_1-e_2}{A}=2$, let $\mathcal{L}:=\mathcal{L}(e_1,e_2)$, and put $G_0=I_A(\mathcal{L})$. 
We will define $P_A^k$ to be of the form $G_0[\mathcal{A}_k]$ for some $\mathcal{A}_k\subset \mathcal{L}$ to be defined inductively.
Let first $s_0=0$, $t_0=e_1$, and $\mathcal{A}_0=\{s_0,t_0\}$. Define $\kappa_0$ to be the length-one path between $s_0$ and $t_0$ in $G_0$.
Suppose inductively that $\mathcal{A}_{k-1}$ has been defined.
Write $t_{k-1}=(r-1)e_1$ for some positive integer $r$ and define $\mathcal{R}_{r}=\{x\in \mathcal{L} \mid d_{G_0}(x,0)=r\}$ and $\mathcal{K}_r=\{x\in G_0\mid d_{G_0}(re_1,x)=1\}$.
Define $\tau_k$ and $\sigma_k$ to be the cycles of $G_0$ through the points of $\mathcal{K}_r$ and $(\mathcal{R}_{r}\setminus \{re_1\})\cup (\mathcal{K}_{r}\setminus \{t_{k-1}\})$, respectively. 
Finally define $\mathcal{T}_r=\{t_{r-1}+i e_1 \mid i=3,\dots ,\ell\}$ where $\ell$ is chosen so large that the path $\kappa_k$ on the vertices of $\mathcal{T}_r$ is so long that it cannot be contained in the cycle $\sigma_k$ in any drawing of $P^k_A$ as an intersection graph of $A$ and $B$. 
Let $t_k=\ell e_1$ and $\mathcal{A}_k=\mathcal{A}_{k-1}\cup (\mathcal{R}_r \setminus \{re_1\})\cup \mathcal{K}_r\cup \mathcal{T}_r$.
See Figure~\ref{fig:christmas}.
\end{construction}

\begin{figure}
\centering
\includegraphics[]{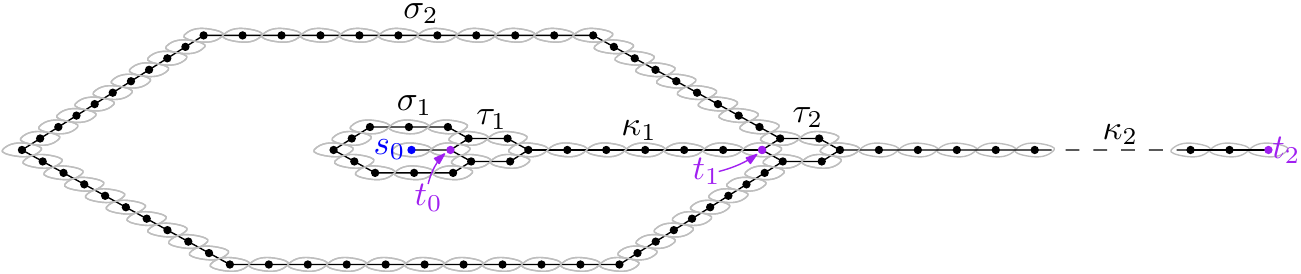}
\caption{The construction of $P_A^2$.}
\label{fig:christmas}
\end{figure}

\begin{lemma}\label{lemma:christmas}
The graph $P_A^k$ has properties~\ref{prop:P1}--\ref{prop:P2}.
\end{lemma}


\begin{proof}

The graph $P_A^0$ trivially has the properties.
Suppose inductively that the $P_A^{k-1}$ has the properties and 
consider any drawing of $P^k_A$ as an intersection graph of $A$ or $B$.
Note that $\kappa_k$ is attached to two cycles: $\tau_k$ and $\sigma_k$.
By construction, $\kappa_k$ is so long that it cannot be contained in any of them.
Hence, $\kappa_k$ is in the exterior of both.
It follows that either $\sigma_k$ is contained in $\tau_k$ or $\tau_k$ is contained in $\sigma_k$.
Clearly, $\tau_k$ is too short for the first to be the case.
We therefore get that all of $P_A^{k-1}$ is contained in $\sigma_k$ and that $\kappa_k$ is on the boundary of the outer face.
Furthermore, the induction hypothesis implies that $s_0$ is contained in the $k$ nested, disjoint cycles $\sigma_1,\ldots,\sigma_k$.
\end{proof}

The most important property of $P^k_A$ is that every vertex $u \in \sigma_k$
has distance $\Omega(k)$ to $s_0$ in any drawing of $P^k_A$ as intersection
graph of $A$ and $B$. This is exactly what we will use when
constructing $Q^k_A$.

\begin{lemma}\label{lem:cycleDist}
Let $X\in\{A,B\}$.
Consider any drawing of $P_A^k$ as an intersection graph of $X$.
For any vertex $u\in\sigma_k$, we have $\norm{s_0u}{X}> 2(k/9-1)$.
\end{lemma}

\begin{proof}
Note that each cycle $\sigma_1,\ldots,\sigma_k$ has an edge that intersects the segment $s_0u$, and that the intersection point has distance at most $1$ to a vertex of the cycle.
We claim that each subsegment $r$ of $s_0u$ of length at most $2$ is intersected by at most $9$ cycles.
Otherwise, the midpoint $x$ of $r$ would have distance at most $2$ to at least $10$ independent vertices.
The translates of $X$ centered at these vertices are pairwise disjoint and contained in a ball $D$ centered at $x$ with radius $3$.
But the area of $D$ is only $3^2=9$ times larger than that of $X$, a contradiction.

Let now $\ell=\norm{s_0u}{X}$, and divide $s_0u$ into $\lceil \ell/2\rceil$ equally long pieces, each of length at most $2$.
As each piece is intersected by at most $9$ cycles, the total number of cycles intersecting $s_0u$ is $9\lceil \ell/2\rceil$.
We get that $k\leq 9\lceil \ell/2\rceil<9(\ell/2+1)$ so that $\ell > 2(k/9-1)$.
\end{proof}

Having defined $P^k_A$ we are now ready to the main part of this section: Constructing
$Q^k_A$ and prove that it has the necessary properties for building $\eps$-overlap
graphs.

\begin{construction}[$Q_A^k$]\label{const:Q}
We here define a graph $Q_A^k\in I(A)$ by specifying a drawing of $Q_A^k$ as an intersection graph of $A$.
Let $k'\mydef 18(k+1)$.
We start with $P_A^{k'}$ and explain what to add to obtain $Q_A^k$.
Let $u_0,\ldots,u_{n-1}$ be the vertices of $\sigma_{k'}$ in cyclic, counter-clockwise order.
Consider an arbitrary drawing of $P_A^{k'}$ as an intersection graph of $A$ and a vertex $u_i$.
Note that $d\mydef \left\lceil\frac{\norm{s_0u_i}{A}-2}{2}\right\rceil$ is the number of vertices needed to add in order to create a path from $s_0$ to $u_i$.
It follows from Lemma~\ref{lem:cycleDist} that $d\geq 2k$.

We want to minimize the vector of these values $d$ with respect to each vertex $u_i\in\sigma_{k'}$.
To be precise, we define
$$(d_0,\ldots,d_{n-1})\mydef
\min\left(
\left\lceil\frac{\norm{s_0u_0}{A}-2}{2}\right\rceil,\ldots,\left\lceil\frac{\norm{s_0u_{n-1}}{A}-2}{2}\right\rceil
\right),$$
where the minimum is with respect to the lexicographical order and taken over all drawings of $P_Q^{k'}$ as an intersection graph.
Consider an drawing of $P_A^{k'}$ as an intersection graph realizing the minimum and let $\mathcal P$ be the set of vertices in the drawing.
For each vertex $u_i$, we create a path $\pi_i$ from $s_0$ to $u_i$ as follows.
Let $\mathbf v_i$ be the unit-vector in direction $u_i-s_0$.
We add new vertices placed at the points $v_i(j)\mydef s_0+2j\mathbf v_i$ for $j\in\{1,\ldots,d_i\}$.
We now define the vertices of $Q_A^k$ as $\mathcal Q\mydef \mathcal P\cup\bigcup_{i=0}^{n-1}\{v_i(1),\ldots,v_i(d_i)\}$ and define $Q_A^k=I_A(\mathcal Q)$.
See Figure~\ref{fig:Q}.
\end{construction}

\begin{figure}
\centering
\includegraphics{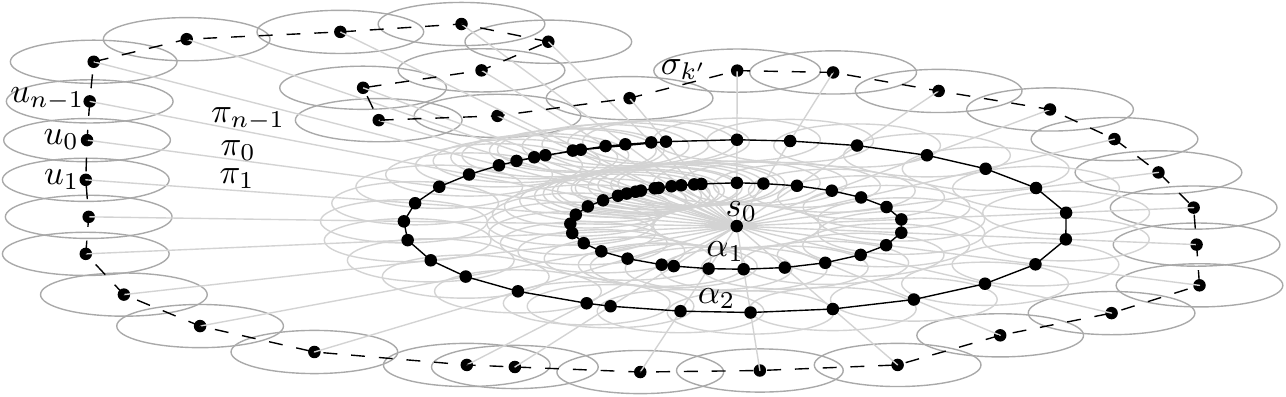}
\caption{A part of a graph $Q_A^k$.
The vertices $v_i(j)$ are only shown for $j\in\{1,2\}$, and only edges on paths $\pi_i$ and cycles $\alpha_1,\alpha_2,\sigma_{k'}$ are shown.}
\label{fig:Q}
\end{figure}

\begin{remark}\label{remark:minimality}
By construction, there exists a drawing of $Q_A^k$ as an intersection graph of $A$.
If there does not exist one with respect to $B$, we are done, since we then clearly have that $I(A) \neq I(B)$.
Now suppose that there exists a drawing of $P_A^{k'}$ as an intersection graph of $B$ such that
\begin{align}\label{remark:minimality:eq}
\left(
\left\lceil\frac{\norm{s_0u_0}{B}-2}{2}\right\rceil,\ldots,\left\lceil\frac{\norm{s_0u_{n-1}}{B}-2}{2}\right\rceil
\right)\prec
(d_0,\ldots,d_{n-1}),
\end{align}
where $\prec$ denotes the lexicographical order.
We can now define a graph $Q_B^k\in I(B)$ in a similar way as we defined $Q_A^k$ by adding $\left\lceil\frac{\norm{s_0u_i}{B}-2}{2}\right\rceil$ vertices to form a path from $s_0$ to each $u_i$.
It then follows from~\eqref{remark:minimality:eq} that $Q_B^k\notin I(A)$, so in this case we have likewise succeeded in proving $I(A)\neq I(B)$.
In the following, we therefore assume that $Q_A^k\in I(B)$ for any $k$ and that no drawing of $P_A^{k'}$ as an intersection graph of $B$ satisfying~\eqref{remark:minimality:eq} exists.
\end{remark}

First we need to show that $Q^k_A$ does contain a cycle $\alpha_k$ as described
in the beginning of this section.

\begin{lemma}\label{lemma:cycleQ}
The set of edges of $Q_A^k$ contain the pairs $v_i(j)v_{i+1}(j)$ for any $i\in[n]$ and $j\in\{1,\ldots,k\}$, and for each $j\in\{1,\ldots,k\}$, these edges thus form a cycle $\alpha_j$.
In the specific drawing of $Q^k_A$ as an intersection graph defined in Construction~\ref{const:Q}, the cycle $\alpha_j$ is contained in the annulus $\setbuilder{x\in\R^2}{\norm{s_0x}{A}\in[2j-1,2j]}$.
\end{lemma}

\begin{proof}
Consider the pair $v_i(j)v_{i+1}(j)$.
Note that $\norm{u_iu_{i+1}}{A}\leq 2$ as $u_iu_{i+1}$ is an edge of $\sigma_{k'}$.
Assume without loss of generality that $\norm{s_0u_i}{A}\leq\norm{s_0u_{i+1}}{A}$, and let $u'_{i+1}$ be the point on $s_0u_{i+1}$ such that $\norm{s_0u'_{i+1}}{A}=\norm{s_0u_i}{A}$.
Note that
$$\norm{s_0u'_{i+1}}{A}+\norm{u'_{i+1}u_{i+1}}{A}=
\norm{s_0u_{i+1}}{A}\leq
\norm{s_0u_i}{A}+\norm{u_iu_{i+1}}{A},$$ so
$\norm{u'_{i+1}u_{i+1}}{A}\leq \norm{u_iu_{i+1}}{A}\leq 2$.
Hence,
$$\norm{u_iu'_{i+1}}{A}\leq
\norm{u_iu_{i+1}}{A}+\norm{u_{i+1}u'_{i+1}}{A}\leq 4.$$
It follows from Lemma~\ref{lem:cycleDist} that $\norm{s_0u_i}{A}\geq 4k$, and thus $\norm{s_0v_i(j)}{A}\leq 2k\leq\norm{s_0u_i}{A}/2$.
As the triangles $s_0u_iu'_{i+1}$ and $s_0v_i(j)v_{i+1}(j)$ are similar, we get
$$\norm{v_i(j)v_{i+1}(j)}{A}\leq \norm{u_iu'_{i+1}}{A}/2\leq2.$$

For the second part, note that the edge $v_i(j)v_{i+1}(j)$ is in the ball $s_0+2jA$.
As any point on $v_0(j)v_{i+1}(j)$ is within distance $1$ from $v_i(j)$ or $v_{i+1}(j)$, the statement follows.
\end{proof}

This shows that the cycle $\alpha_k$ behaves nicely in one particular drawing
of $Q^k_A$ as an intersection graph.
We now show that something similar holds for every drawing.
\iffull
\begin{lemma}\label{lemma:annulus}
Let $X\in\{A,B\}$.
In any drawing of $Q^k_A$ as an intersection graph with respect to $X$, any $i\in[n]$, and any $j\in\{1,\ldots,k\}$, the vertex $v_i(j)$ is contained in the annulus $\setbuilder{x\in\R^2}{\norm{s_0x}{A}\in(2j-2,2j]}$.
Therefore, the cycle $\alpha_j$ is contained in the annulus $\setbuilder{x\in\R^2}{\norm{s_0x}{A}\in(2j-3,2j]}$.
\end{lemma}

\begin{proof}
The upper bound on $\norm{s_0v_i(j)}{X}$ holds as there is a path from $s_0$ to $v_i(j)$ consisting of only $j+1$ vertices.
For the lower bound, assume for contradiction that for some values of $i$ and $j$, there exists a drawing of $Q^k_A$ as an intersection graph where $\norm{s_0v_i(j)}{X}\leq 2j-2$.
We now claim that (i): $\left\lceil\frac{\norm{s_0u_j}{X}-2}{2}\right\rceil\leq d_j$ for all $j\in\{1,\ldots,k\}$ and (ii): $\left\lceil\frac{\norm{s_0u_i}{X}-2}{2}\right\rceil< d_i$. Together, (i) and (ii) contradict either the minimality of $(d_0,\ldots,d_{n-1})$ (if $X=A$) or the assumption from Remark~\ref{remark:minimality}~(if $X=B$).

For part (i), note that since $\pi_j$ consists of $d_j+2$ vertices, we have $\norm{s_0u_j}{X}\leq 2(d_j+1)$, and it follows that $\left\lceil\frac{\norm{s_0u_j}{X}-2}{2}\right\rceil\leq \left\lceil d_j\right\rceil=d_j$.

For part (ii), note that since there is a path from $v_i(j)$ to $u_i$ consisting of $d_i-j+2$ vertices, we have $\norm{v_i(j)u_i}{X}\leq 2(d_i-j+1)$.
By the triangle inequality we now get
$\norm{s_0u_i}{X}\leq \norm{s_0v_i(j)}{X}+\norm{v_i(j)u_i}{X}\leq 2j-2+2(d_i-j+1)=2d_i$.
It now follows that $\left\lceil\frac{\norm{s_0u_i}{X}-2}{2}\right\rceil\leq\left\lceil\frac{2d_i-2}{2}\right\rceil=d_i-1<d_i$.
\end{proof}

We want to be able to conclude that two cycles $\alpha_k, \alpha'_k$ intersect
if the two annuli containing the cycles cross each other. This will be an easy consequence of the following lemma
which shows that the cycle $\alpha_k$ goes all the way around $s_0$ inside the annulus
in any drawing of $Q^k_A$ as an intersection graph.
\begin{lemma}\label{lemma:winding}
Let $X\in\{A,B\}$.
Consider any drawing of $Q_A^k$ as an intersection graph with respect to $X$ and any $j\in\{3,\ldots,k\}$, and consider $\alpha_j$ as a parameterized, closed curve $\alpha_j\colon [0,1]\longrightarrow \R^2$, such that for any $i\in[n]$, $\alpha_j$ interpolates linearly from $v_i(j)$ to $v_{i+1}(j)$ on the interval $[\frac{i}{n},\frac{i+1}{n}]$, where indices are taken modulo $n$.
We may then define a continuous argument function $\theta_{\alpha}\colon[0,1]\longrightarrow \R$ such that $\theta_{\alpha}(t)$ is an argument of the vector $\alpha_j(t)-s_0$ for all $t\in[0,1]$.
Then the argument variation of $\alpha_j$ around $s_0$ is $\theta_{\alpha}(1)-\theta_{\alpha}(0)=\pm 2\pi$.
\end{lemma}

\begin{proof}
Just as $\alpha_j$ is considered as a parameterized curve in the lemma, we may in a similar way consider $\sigma_{k'}$ as a parameterized closed curve $\mathcal \sigma_{k'}\colon [0,1]\longrightarrow \R^2$ such that $\sigma_{k'}$ interpolates linearly from $u_i$ to $u_{i+1}$ on the interval $[\frac{i}{n},\frac{i+1}{n}]$.
We also define $\theta_\sigma\colon[0,1]\longrightarrow\R$ to be a continuous argument function for $\sigma_{k'}$.
Since $\sigma_{k'}$ is a simple closed curve containing $s_0$ in the interior by Lemma~\ref{lemma:christmas}, we get that the argument variation of $\sigma_{k'}$ around $s_0$ is $\theta_\sigma(1)-\theta_\sigma(0)=2\pi$.

For any $\lambda\in[0,1]$, we now define the curve $\varphi_\lambda\colon[0,1]\longrightarrow\R^2$ such that $\varphi_\lambda(t)=(1-\lambda)\alpha_j(t)+\lambda \sigma_{k'}(t)$.
Thus, $\varphi_\lambda$ is a continuous interpolation between $\alpha_j$ (when $\lambda=0$) and $\sigma_{k'}$ (when $\lambda=1$).

We claim that for all $t,\lambda\in[0,1]$, we have $s_0\neq\varphi_\lambda(t)$.
To this end, we prove that the segment $\alpha_j(t)\sigma_{k'}(t)=\varphi_0(t)\varphi_1(t)$ is contained in the ball $D\mydef \setbuilder{x\in\R^2}{\norm{xu_i}{X}\leq \norm{\alpha_j(t)u_i}{X}}$, whereas $s_0\notin D$.

Suppose that $t=i/n+t'$, where $t'\in[0,\frac 1{2n})$.
The case where $t'\in[\frac 1{2n},\frac 1n)$ is similar.
We now have that
$$\norm{\alpha_j(t)u_i}{X}\leq
\norm{\alpha_j(t)v_i(j)}{X}+
\norm{v_i(j)u_i}{X}\leq
1+2(d_i-j+1)<2(d_i-1)+1<\norm{s_0u_i}{X},$$
so $s_0\notin D$.
Obviously, $\alpha_j(t)\in D$ by definition.
Since also $\norm{\sigma_{k'}(t)u_i}{X}\leq 1<\norm{s_0u_i}{X}$, the claim follows.

It now follows that for all $\lambda\in[0,1]$, the curve $\varphi_\lambda$ has the same argument variation around $s_0$ as $\sigma_{k'}$.
In particular, $\theta_{\alpha}(1)-\theta_{\alpha}(0)=\pm 2\pi$ as stated.
\end{proof}

\begin{figure}
\centering
\includegraphics{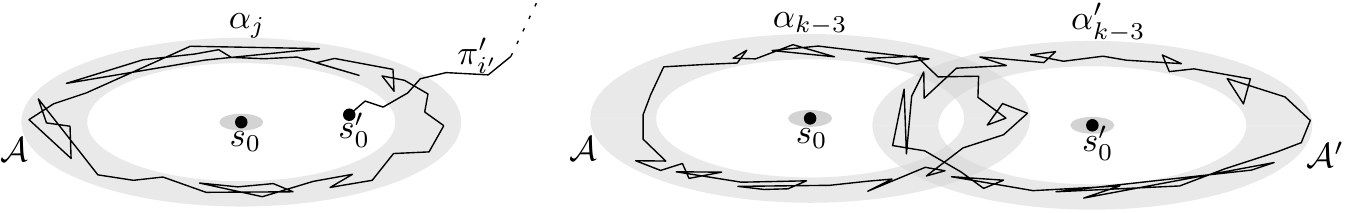}
\caption{Two cases from the proof of Lemma~\ref{lemma:realizeH}.}
\label{fig:annulus}
\end{figure}
\fi

We will now use use copies of $Q^k_A$ to construct $\eps$-overlap graphs which will finish the proof.

\begin{construction}[$H_A^k(G)$]\label{const:H}
For any $G\in C(A)$, consider a fixed drawing of $G$ as a contact graph of $A$.
For each vertex $w$ of $G$, we make a copy of the drawing of $Q_A^k$ as an intersection graph as defined in Construction~\ref{const:Q} which we translate so that $s_0$ is placed at $s_0^w\mydef (2k-2)w$.
We then add all edges induced by the vertices, and the result is denoted as $H_A^k(G)$.
\end{construction}

To show that $H^k_A(G)$ does in fact construct $G$ as a $\eps$-overlap graph, we need to show that if $w w'$ is an edge of $G$, then the corresponding cycles $\alpha_k$ and $\alpha'_k$ intersect each other.
That implies the existence of vertices $v_i(k)$ and $v'_{i'}(k)$ of $\alpha_k$ and $\alpha'_k$, respectively, such that $v_i(k)v'_{i'}(k)$ is an edge of $H_A^k(G)$.

\begin{lemma}\label{lemma:edgesH}
Consider two vertices $w,w'$ of a drawing of a graph $G$ as a contact graph.
Denote by $Q$ and $Q'$ the copies of $Q_A^k$ in $H_A^k(G)$ corresponding to $w$ and $w'$, respectively, such that $s_0,\pi_i,\alpha_j,v_i(j)$ denote objects in $Q$ and $s'_0,\pi'_i,\alpha'_j,v'_i(j)$ denote objects in $Q'$.
If $v_i(j)v'_{i'}(j')$ is an edge of $H_A^k(G)$, then $j+j'\geq 2k-4$.

If $ww'$ is an edge of $G$, then there is an edge $v_i(k)v_{i'}'(k)$ in $H_A^k(G)$.
\end{lemma}

\begin{proof}
Assume that $j+j'\leq 2k-3$.
In the drawing of $H^k_A(G)$ as an intersection graph defined by Construction~\ref{const:H}, we have
$$\norm{v_i(j)v'_{i'}(j')}{A}\geq \norm{s_0s'_0}{A}-\norm{s_0v_i(j)}{A}-\norm{s'_0v'_{i'}(j')}{A}= 4k-4-2j-2j'\geq 4>2,$$
so $v_i(j)v'_{i'}(j')$ is not an edge of $H_A^k(G)$.

Suppose now that $ww'$ is an edge of $G$.
Let $\mathcal A\mydef \setbuilder{x\in\R^2}{\norm{s_0x}{A}\in[2k-1,2k]}$ be the annulus containing $\alpha_k$ (by Lemma~\ref{lemma:cycleQ}) and $\mathcal A'\mydef \setbuilder{x\in\R^2}{\norm{s'_0x}{A}\in[2k-1,2k]}$ be that containing $\alpha'_k$.
The annuli $\mathcal A$ and $\mathcal A'$ cross over each other as two Olympic rings (i.e., the difference $\mathcal A\setminus\mathcal A'$ has two connected components).
\iffull
\Cref{lemma:winding} then shows the
\fi
intuitive fact that $\alpha_k$ and $\alpha'_k$ must intersect.
Therefore, there is an edge $v_i(k)v'_{i'}(k)$.
\end{proof}

We need one last fact before concluding that $H^k_A(G)$ constructs $G$ as an
$\eps$-overlap graph:
for $X\in\{A,B\}$ and any two centers $s_0, s'_0$ of copies
of $Q^k_A$, we have $\norm{s_0s'_0}{X}=\Omega(k)$, and if the corresponding
vertices $w, w'$ share an edge in $G$, then $\norm{s_0s'_0}{X}=O(k)$, as made precise in the following
\iffull
lemma.
\fi

\begin{lemma}\label{lemma:realizeH}
Let $X\in\{A,B\}$ and $k\geq 7$.
In the setting of Lemma~\ref{lemma:edgesH}, consider an arbitrary drawing of $H_A^k(G)$ as in intersection graph of $X$.
Then $\norm{s_0s'_0}{X}\geq 4k-18$ and if $ww'$ is an edge of $G$, then $\norm{s_0s'_0}{X}\leq 4k+2$.
\end{lemma}

\iffull
\begin{proof}
    To prove the lower bound on $\norm{s_0s'_0}{X}$, we show that $\norm{s_0s'_0}{X}$ is neither in the interval $[0,k-7)$ nor in $[k-7,4k-18)$.
    The cases are depicted in Figure~\ref{fig:annulus}.
    \begin{enumerate}
    \item\label{lemma:realizeHcase1}
        Suppose that $\norm{s_0s'_0}{X}\in [0,k-7)$.
        Let $j\mydef \lceil\frac k2-2\rceil$.
        Then $s'_0$ is in the ball $s_0+(2j-3)X$.
        By Lemma~\ref{lemma:annulus}, $\alpha_j$ is contained in the annulus $\mathcal A\mydef \setbuilder{x\in\R^2}{\norm{s_0x}{X}\in (2j-3,2j]}$.
        Note that $\mathcal A$ has diameter $4j<2k-3$.
        It follows from Lemma~\ref{lemma:annulus} that any of the subpaths of $\pi'_{i'}$ from $s'_0$ to $v'_{i'}(k)$ connects the inner and outer boundary of $\mathcal A$.
        Therefore, Lemma~\ref{lemma:winding} gives that $\pi'_{i'}$ crosses $\alpha_j$.
        Thus, there is also an edge $v_i(j)v'_{i'}(j')$ of $H_A^k(G)$, where $j+j'\leq \lceil\frac k2-2\rceil+k<2k-4$, contradicting Lemma~\ref{lemma:edgesH}.
    
    \item\label{lemma:realizeHcase2}
        Suppose now that $\norm{s_0s'_0}{X}\in [k-7,4k-18)$.
        By an argument similar to the one used in the proof of Lemma~\ref{lemma:edgesH}, one can show that $\alpha_{k-3}$ crosses $\alpha'_{k-3}$.
        Hence, there is an edge $v_i(k-3)v'_{i'}(k-3)$, contradicting Lemma~\ref{lemma:edgesH}.
    \end{enumerate}
    
    Consider now the case that $ww'$ is an edge of $G$.
    By Lemma~\ref{lemma:edgesH}, we know that there is an edge $v_i(k)v'_{i'}(k)$.
    Then $\norm{s_0s'_0}{X}\leq \norm{s_0v_i(k)}{X}+\norm{v_i(k)v'_{i'}(k)}{X}+\norm{v'_{i'}(k)s'_0}{X}\leq 4k+2$.
\end{proof}
\else
\fi

We are now ready to show that $H^k_A(G)$ does in fact construct $G$ as a
$\eps$-overlap graph.
\begin{lemma}\label{lemma:overlap}
For a graph $G=(V,E)\in C(A)$, if $H_A^k(G)\in I(B)$ for $k\geq 7$, then $G\in C_{10/k}(B)$.
\end{lemma}

\begin{proof}
Suppose that $H_A^k(G)\in I(B)$ and consider a drawing of $H_A^k(G)$ as an intersection graph of $B$, and define $\mathcal A\mydef \setbuilder{\frac{s_0^u}{2k+1}}{u\in V}$.
It follows directly from Lemma~\ref{lemma:realizeH} that $I_B(\mathcal A)$ is a drawing of $G$ as a $\left(2-\frac{4k-18}{2k+1}\right)$-overlap graph of $B$. 
Since $\frac{4k-18}{2k+1}\geq 2-10/k$, the statement follows.
\end{proof}

\Cref{theorem:main} is  an easy consequence of \Cref{lemma:overlap} and \Cref{lemma:implyContact}:
\begin{proof}[Proof of \Cref{theorem:main}]
Let $G=(V,E')$ have the property from the theorem and suppose that $I(A)=I(B)$.
Then in particular, $H_A^k(G)\in I(B)$ for all $k>0$.
By Lemma~\ref{lemma:overlap} and~\ref{lemma:implyContact}, there is a graph $H=(V,E)\in C(B)$ such that $E'\subseteq E$, which is a contradiction.
\end{proof}


	\section*{Acknowledgement}

We thank Tillmann Miltzow for asking when the translates of two different convex bodies induce the same intersection graphs which inspired us to work on these problems.

	\bibliographystyle{plain}
	\bibliography{lib}

\end{document}